\documentclass[a4paper,UKenglish,cleveref, autoref, thm-restate]{lipics-v2021}

\usepackage{xcolor}
\usepackage{amssymb,amsmath,amsthm}
\usepackage{mathtools}
\usepackage{tcolorbox}

\usepackage[linesnumbered, ruled, vlined]{algorithm2e}
\SetKwRepeat{Do}{do}{while}%

\SetCommentSty{mycommfont}

\usepackage[]{todonotes}
\usepackage{bbm}
\usepackage{wrapfig}
\usepackage{subcaption}
\usepackage{stmaryrd}

\usepackage{MnSymbol}

\usepackage{multirow}

\usepackage{colortbl}

\usepackage{graphicx}

\usepackage{tikz}
\usetikzlibrary{matrix}

\newcommand{\FOL}{\text{FO}}

\newcommand{\timeDef}{\textit{def}}

\newcommand{\FOLOrder}{{\FOL[<]}}

\newcommand{\TFOL}{{\FOL[<,\allSetTraces]}}

\newcommand{\TraceTFOL}{\setsetTraces\text{-}\TFOL}

\newcommand{\TimeTFOL}{\mathord{<}\text{-}\TFOL}

\newcommand{\GFOL}{\Globally\text{-}\TFOL}

\newcommand{\modelsTFOL}{\models_{\allSetTraces}}

\newcommand{\timeVar}{i}
\newcommand{\timeVars}{\mathcal{I}}

\newcommand{\SetM}{\textbf{M}}
\newcommand{\M}{M}

\newcommand{\ind}[3]{\Ind_{#3}{(#1,#2)}}
\newcommand{\Ind}{\textit{ind}}

\newcommand{\sync}{\textit{sync}}
\newcommand{\async}{\textit{async}}
\newcommand{\novisible}{\textit{hidden}}
\newcommand{\pointW}{\textit{point}}
\newcommand{\seg}{\textit{seg}}

\newcommand{\syncS}{\pointW}
\newcommand{\syncO}{\seg}

\newcommand{\Prop}{X}
\newcommand{\trace}{\tau}

\newcommand{\val}{v}
\newcommand{\vals}{\mathbb{V}}

\newcommand{\Until}{\mathbin\mathbf{U}}
\newcommand{\Next}{\mathop\mathbf{X}}
\newcommand{\Globally}{\mathop\mathbf{G}}
\newcommand{\Eventually}{\mathop\mathbf{F}}

\newcommand{\generatedSet}[1]{\llbracket #1 \rrbracket}
\newcommand{\setTraces}{T}
\newcommand{\allSetTraces}{\mathbb{T}}
\newcommand{\system}{S}
\newcommand{\setsetTraces}{\textbf{T}}

\newcommand{\traceVar}{\pi}

\newcommand{\traceAssign}{\Pi}
\newcommand{\emptyAssign}{\traceAssign^{\emptyset}}
\newcommand{\flatT}[1]{\langle  #1 \rangle}

\newcommand{\modelsHyper}{\models_{H}}
\newcommand{\notmodelsHyper}{\not\models_{H}}

\newcommand{\Var}{\mathcal{V}}
\newcommand{\fr}{\text{free}}

\newcommand{\LTLC}{\mathbb{C}}
\newcommand{\LTLX}{\mathbb{X}}

\newcommand{\LTLBox}{\mathbb{G}}

\newcommand{\HyperLTLC}{2^{\LTLC}}

\newcommand{\Equiv}{\approx}

\newcommand{\wit}{f}

\newcommand{\pIn}{\textbf{ in }}

\newcommand{\pTrue}{\textbf{True}}
\newcommand{\pInput}{\textbf{input}}
\newcommand{\pOutput}{\textbf{output}}

\newcommand{\pDefault}{\textit{default}}

\newcommand{\tAnd}{\text{ and }}
\newcommand{\tOr}{\text{ or }}

\newcommand{\tIf}{\text{ if }}

\newcommand{\tIff}{\text{ iff }}

\newcommand{\nat}{\mathbb{N}}

\newcommand{\state}{\emph{state}}

\newcommand{\DDD}{\mathord{\ldots}}
\nolinenumbers

\title{Flavours of Sequential Information Flow} 

\titlerunning{Flavours of Sequential Information Flow} 

\author{Ezio Bartocci}{Technische Universit\"at Wien, Vienna, Austria \and \url{http://www.eziobartocci.com} }{ezio.bartocci@tuwien.ac.at}{https://orcid.org/0000-0002-8004-6601}{}

\author{Thomas Ferr\`ere}{Imagination Technologies, Kings Langley, UK }{thomas.ferrere@imgtec.com}{https://orcid.org/0000-0001-5199-3143}{}

\author{Thomas A. Henzinger}{IST Austria, Klosterneuburg, Austria \and \url{http://pub.ist.ac.at/~tah/} }{tah@ist.ac.at}{https://orcid.org/0000-0002-2985-7724}{}

\author{Dejan Nickovic}{AIT Austrian Institute of Technology, Vienna, Austria }{dejan.nickovic@ait.ac.at}{https://orcid.org/0000-0001-5468-0396}{}

\author{Ana Oliveira da Costa}{Technische Universit\"at Wien, Vienna, Austria }{ana.costa@tuwien.ac.at}{https://orcid.org/0000-0002-8741-5799}{This work was supported by the Austrian FWF project W1255-N23}

\authorrunning{E. Bartocci et al.} 

\Copyright{E. Bartocci and T. Ferr\`ere and T. A. Henzinger and D. Nickovic and A. Oliveira da Costa} 



\category{} 

\relatedversion{} 

\EventEditors{}
\EventNoEds{0}
\EventLongTitle{}
\EventShortTitle{}
\EventAcronym{}
\EventYear{2021}
\EventDate{}
\EventLocation{}
\EventLogo{}
\SeriesVolume{}
\ArticleNo{}

\begin{document}

\maketitle

\begin{abstract}
Information-flow policies prescribe which information is available to a given user or subsystem. 
We study the problem of specifying such properties in reactive systems, which may require dynamic changes in information-flow restrictions between their states.
We formalize several flavours of \emph{sequential information-flow}, which cover different assumptions about the semantic relation between multiple observations of a system.  
Information-flow specification falls into the category of \emph{hyperproperties}. 
We define different variants of sequential information-flow specification using a first-order logic with both trace quantifiers and temporal quantifiers called Hypertrace Logic.
We prove that HyperLTL, equivalent to a subset of Hypertrace Logic with restricted quantifier prefixes, cannot specify the majority of the studied 
two-state independence variants.  
For our results, we introduce a notion of equivalence between sets of traces that cannot be distinguished by certain classes of formulas in Hypertrace Logic. 
This presents a new approach to proving inexpressiveness results for logics such as HyperLTL.
\end{abstract}

\section{Introduction}
\label{sec:intro}

Information-flow policies specify restrictions on what information can be shared within components of a system or its users. Information that must be kept secret may be deduced by combining multiple 
observations of the non-secret behavior of the system. For this reason, properties that characterize information-flow policies are often not properties of a single trace, but rather properties of sets of traces, that is, \emph{hyperproperties} \cite{ClarksonS10}.

A basic concept for specifying information flows can be found in the notion 
of {\em independence} \cite{gradel2013dependence, ClarksonS10}, defined as a binary relation between observable variables of a system.
We say that $y$ is independent of $x$, denoted by \(\ind{x}{y}{}\), 
to specify that no information can flow from $x$ to $y$. Over a given set of traces, 
the independence relation \(\ind{x}{y}{}\) is captured by the formula \(\forall \traceVar \forall \traceVar' \exists \traceVar''\ x_{\traceVar} = x_{\traceVar''} \wedge y_{\traceVar'} = y_{\traceVar''}\), where 
$\pi$ denotes an observation of the system and $x_{\pi}$ the value $x$ observed in $\pi$.
We introduce {\em two-state independence}, a simple, yet fundamental sequential information-flow requirement. 
It can be used for instance to capture {\em declassification}~\cite{sabelfeld2009declassification}, 
a process in which previously secret information is allowed to be released.
Given $x$, $y$, $z$ and $\state$ variables, the two-state independence is stated as follows:
``The value of \(y\) is independent from the value of \(x\) until 
$\state$ changes,
and from then on the value of \(z\) is independent from the value of \(x\).''


The program $\mathcal{P}$, shown in 
Algorithm~\ref{alg:pk}, intuitively satisfies a two-state independence property between $x$, $y$ and $z$.
The program starts in the initial 
state ($\state = 0$) and in every subsequent step, the next state is 
non-deterministically 
assigned via the channel $c_1$. Once $\mathcal{P}$ changes from 
$\state=0$ to $\state=1$, it remains in that state forever. 
The value of $x$ is non-deterministically assigned via 
channel $c_0$ regardless of the current state. When in state 
$0$, $\mathcal{P}$ assigns $x$ to $z$ and a default 
value to $y$. When in state 
$1$, $\mathcal{P}$ assigns $y$ to $z$ and a default value 
to $z$. Note that the default value can be $0$, $1$ or a non-deterministic 
boolean value set at 
the start of the program execution.
The program finally exposes $y$ and $z$ via 
channels $c_2$ and $c_3$, respectively. 
Program $\mathcal{P}$ satisfies the two-state independence 
requirement by ensuring that \(\ind{x}{y}{}\) holds in the first state, and that \(\ind{x}{z}{}\) holds in the second state. 
Table~\ref{tab:ex:motivation} 
shows a set of traces observed from the input/output interface of $\mathcal{P}$ 
and that are consistent with the two-state independence requirement. 
The first two traces, \(\trace_1\) and \(\trace_2\), transition to the second state at time 1, while \(\trace_3\) and \(\trace_4\) transition at time 2 and 3, respectively. Then, for the second state of the specification (i.e.\ after \(\state=1\)), we need to compare the observations at time 1 of \(\trace_1\) and \(\trace_2\) with
observations at time 2 and 3 of \(\trace_3\) and \(\trace_4\), respectively.

\begin{table}
\begin{minipage}{0.4\textwidth}
\begin{algorithm}[H]
\small
\SetAlgoLined
\SetKwRepeat{Do}{do}{while}
 \(\state:=0\)\;
 \Do{\(\pTrue\)}{
  \If{\em (\( \state = 0\))}
  {\(\pInput(c_{1}, \state \pIn \{0,1\})\);}
  \(\pInput(c_{0}, x \pIn \{0,1\})\);\\
  \eIf{\em (\(\state = 0\))}
  {\(z := x\); \(y = \pDefault\);
  }
  {\(y := x\); \(z = \pDefault\);}
  \(\pOutput(c_{2}, y)\)\;
  \(\pOutput(c_{3}, z)\)\;
 }
 \caption[Program]{Program \(\mathcal{P}\) for two-state independence.}
 \label{alg:pk}
\end{algorithm}
\end{minipage} \quad
\begin{minipage}{0.56\textwidth}
\caption{A set of traces $T$ observed from the inputs and outputs of $\mathcal{P}$, with \(\pDefault=0\), where white cells denote that $\state =0$ and gray cells denote that $\state=1$. 
\label{tab:ex:motivation}}
\centering
\small 
\begin{tabular}{c|ccc|ccc|ccc|ccc}
\multicolumn{1}{c}{}& \multicolumn{12}{c}{Time}\\
\multicolumn{1}{c}{}&\multicolumn{3}{c}{0}&\multicolumn{3}{c}{1}&\multicolumn{3}{c}{$2$} & \multicolumn{3}{c}{$3$} \\
&x & y & z & x & y & z & x & y & z & x & y & z\\ 
\hline 
\(\trace_1\)&0 & 0 & 0 & 
\cellcolor{lightgray}1 & \cellcolor{lightgray}1 & \cellcolor{lightgray}0 & \cellcolor{lightgray}1 & \cellcolor{lightgray}1 & \cellcolor{lightgray}0 &
\cellcolor{lightgray}1 & \cellcolor{lightgray}1 & \cellcolor{lightgray}0 \\
\(\trace_2\)& 1 & 0 & 1 & 
\cellcolor{lightgray}1 & \cellcolor{lightgray}1 & \cellcolor{lightgray}0 &
\cellcolor{lightgray}1 & \cellcolor{lightgray}1 & \cellcolor{lightgray}0 &
\cellcolor{lightgray}1 & \cellcolor{lightgray}1 & \cellcolor{lightgray}0\\
\(\trace_3\)&1 & 0 & 1 & 1 & 0 & 1 & \cellcolor{lightgray}0 & \cellcolor{lightgray}0 & \cellcolor{lightgray}0 & \cellcolor{lightgray}0 & \cellcolor{lightgray}0 & \cellcolor{lightgray}0 \\
\(\trace_4\)&0 & 0 & 0 & 1 & 0 & 1 & 0 & 0 & 0 & \cellcolor{lightgray}1 & \cellcolor{lightgray}1 & \cellcolor{lightgray}0\\
\end{tabular}
\vspace{2cm}
\end{minipage}
\end{table}

Table~\ref{tab:ex:motivation} illustrates one possible way to observe 
the program $\mathcal{P}$. However, the observation of 
program executions may not be 
uniquely defined. For example, an observer may have the ability to access the 
internal program memory, while another observer may only observe 
its input/output interface. The power of observer has a significant 
impact on the specification of information flow requirements.

In this paper, we study multiple flavours of sequential information 
flow, according to our assumptions about the observer. We focus on 
the two-state independence requirement as the simplest hyperproperty 
that exposes the main features of sequential information flow. 

Logical specification of sequential information flow (and other hyperproperties) 
requires (implicit or explicit) quantification over time and traces. We refer to this family of linear-time specification languages as {\em hyperlogics}.
We introduce Hypertrace Logic, a two-sorted first-order logic that allows us to express and compare 
a rich variety of sequential hyperproperties and specification languages for hyperproperties.

We identify two natural interpretations of two-state independence, 
based on {\em point} and {\em segment} semantics. 
In point semantics, an observation at a given execution point is 
independent from observations at all other execution points. 
Segment semantics relates entire segments of observations, where each segment 
is aligned with a specification state.
We also identify 
three types of state transition actions,  
{\em synchronous}, {\em asynchronous} and {\em hidden}.
For example, the set of traces $T$ shown in Table~\ref{tab:ex:motivation} has an asynchronous action and satisfies a two-state independence under the segments semantics, with \(\ind{x}{y}{}\) and \(\ind{x}{z}{}\) interpreted over segments of traces associated to states 0 and 1 respectively.
Every combination of independence interpretation and action 
type defines a different assumption about program observations.
We provide a mathematical definition using Hypertrace Logic of the 
two-state independence for each such combination.

We then study the expressiveness of all the presented two-state independence flavours with respect to HyperLTL, the de-facto standard for specifying, analysing and 
model-checking hyperproperties. We show that HyperLTL cannot express the 
majority of the studied two-state independence variants. 
Our results emphasize the important role 
that the order of time and trace quantifiers play in 
hyperproperties and in addition highlight the need to 
explore, also noted independently in \cite{lics2021,cav2021}, 
more asynchronous variants of hyperlogics.

\noindent The contributions of this paper can be summarized as follows:
\begin{itemize}
    \item We investigate multiple flavours of sequential information flow through a generic first-order formalism, which relieves us of the burden of specific syntactic choices.
    \item We present a comprehensive expressiveness study of the simplest sequential information-flow property---namely, two-state independence---with 
    respect to first-order fragments and the popular HyperLTL formalism. 
    \item We devise a new 
    systematic technique to prove that logics such as to HyperLTL cannot express a given property. 
    This proof strategy is of independent interest and can be used in other expressiveness proofs.
\end{itemize}


\section{First Order Logic for Trace Sets}
\label{sec:FOL:TraceSets}

We define a two-sorted first-order logic to formalize the hyperproperties we are interested in.
We extend the first-order logic of linear order with equality, \(\FOLOrder\) \cite{Kamp1968}, with a trace sort \(\allSetTraces\). 
As we are interested in discrete linear-time,  
we interpret \(\FOLOrder\) with the theory of natural numbers.
Under this theory, \(\FOLOrder\) is expressively equivalent to LTL \cite{Kamp1968,gabbay1980temporal}.

Let \(\Prop\) be a finite set of propositional variables. We denote by \(\val(x)\)  a \emph{valuation} (partial mapping) of  variables \(x \in \Prop\) to boolean 
values, \(v: \Prop \rightarrow \{0,1\}\),
and by \(\vals_{\Prop}\) the set of all valuations over \(\Prop\).
The \emph{domain} of a valuation \(\val\) is denoted as \(\Prop(\val)\) and its \emph{size} is defined by the size of its domain, i.e.\
\(|\val| = |\Prop(\val)|\).
Given a sequence of propositional variables \((x_0, \ldots, x_n)\), we write a valuation \(\val\) over it as a boolean string \(\val(x_0)\ldots \val(x_n)\).
We denote by \(\val[x \mapsto b]\) the update of valuation \(\val\) with \(x\) being assigned the boolean \(b\).
The \emph{composition of two valuations} \(\val\) and \(\val'\) is defined as 
\({\val \otimes \val'} = \val[x_1 \mapsto \val'(x_1)]\ldots[x_n \mapsto \val'(x_n)]\),  where  \(\{x_1, \ldots, x_n\} = \Prop(\val')\).

A trace \(\trace\) over \(\Prop\)  is a sequence of valuations in \(\vals_{\Prop}\). 
We refer to \(\Prop\) as the \emph{alphabet of} \(\trace\).
The set of all \emph{infinite traces (over \(\Prop\))} is denoted by \(\vals_{\Prop}^\omega\) and the set of all \emph{finite traces} is denoted by \(\vals_{\Prop}^*\). For a finite trace \(\tau = v_0 v_1 \dots v_n\), its \emph{length} is defined as \(|\tau| = n + 1\) and \(|\tau| = \omega\) for an infinite trace.
 The \emph{composition} of traces \(\trace = \val_0 \val_1 \ldots \) and \(\trace' = \val'_0 \val'_1 \ldots \) is defined as 
\(\trace \otimes \trace' = (v_0\otimes v'_0) (v_1\otimes v'_1) \ldots \).
Given a trace \(\trace = v_0 v_1 \dots\) and an index \(i < |\trace|\), we use the following indexing notation: \(\trace[i] = v_i\), \(\trace[i\ldots] = v_i v_{i+1}\dots\), and \(\trace[\ldots i] = v_0 v_1\dots v_{i-1}\). For \(j \geq |\trace|\) we adopt the following convention: \(\trace[j\ldots]\) is the empty trace  and \(\trace[\ldots j] = \trace\).

A \emph{trace property} \(\setTraces\) over a set of propositional variables \(\Prop\) is a set of infinite traces over \(\Prop\), that is, \(\setTraces \subseteq \vals^\omega_{\Prop}\). 
The set \(\allSetTraces = 2^{\vals^\omega_{\Prop}}\) defines the \emph{set of all trace properties}.
A system \(\system\) is characterized by the observable behavior for each of its executions, which are represented as traces. Hence a system is defined by a set of traces.
A \emph{hyperproperty} characterizes a set of systems, 
and defines a set of sets of traces \(\setsetTraces \subseteq 2^{\vals^\omega_{\Prop}} = \allSetTraces\).


LTL is a propositional linear-time temporal logic \cite{Pnueli77}.
Its formulas, \(\varphi\), are defined by the grammar: \(\varphi ::= \ a \, |\, \neg \varphi \, |\, \varphi \vee \varphi \, | \, \Next \varphi \, | \, \varphi \Until \varphi\), where \(a\in \Prop\) is a propositional variable and  \emph{next}, \(\Next\), and \emph{until}, \(\Until\), are temporal modalities.
LTL formulas are interpreted over infinite traces. The satisfaction relation, for a given trace \(\trace \in \vals^{\omega}_{\Prop}\), is defined inductively over LTL formulas as follows:
\[
\begin{split}
&\trace \models a \tIff \tau[0](a) = 1; 
\trace \models \neg  \psi \tIff \trace \not\models  \psi; 
  \trace \models  \psi_1 \vee  \psi_2 \tIff 
\trace \models  \psi_1 \tOr \trace \models  \psi_2;
\trace \models \Next  \psi \tIff  \trace[1\ldots] \models  \psi;\\
&\trace \models  \psi_1 \Until  \psi_2 \tIff 
\text{there exists } 0 \leq j:\ \trace[j\ldots] \models  \psi_2 \tAnd
\text{for all } 0 \leq j' < j: \ \trace[j' \ldots] \models  \psi_1.
\end{split}
\]
The temporal operators \emph{globally}, \(\Globally\), and \emph{eventually}, \(\Eventually\) are defined as customary, with  \({\Globally \psi  \equiv \psi \Until \text{false}}\) and \({\Eventually \psi  \equiv \text{true} \Until \psi}\).

\subsection{Hypertrace Logic}

\emph{Hypertrace Logic}, denoted \(\TFOL\), is a two-sorted first-order logic with equality and the signature \({\{<\}} \cup {\{P_a \ |\ a \in \Prop \}} \cup \{\timeDef\}\), where 
\(\Prop\) is a set of propositional variables. It includes the trace sort \(\allSetTraces\) and time sort \(\nat\). 
All the predicates are binary and they have the following signatures: \(\mathord<: \nat \times \nat\) and \(P_a, \timeDef: \allSetTraces \times \nat\), for all \(a\in \Prop\). The predicate \(<\) is interpreted over the theory of natural numbers, while the other predicates are uninterpreted.

The first-order logic of linear order, \(\FOLOrder\), allows only monadic predicates, aside from the interpreted binary predicate \(<\). 
We extend it to specify hyperproperties by allowing binary predicates \(P_a\) for each propositional variable \(a\) and a binary predicate \(\timeDef\).
Given a set of traces, we interpret \(P_a\) with all pairs of traces and time positions where \(a\) holds. The predicate \(\timeDef\) holds for all positions that are within the length of a given trace. This enables us to reason about both finite and infinite traces.

Given a set of traces \(\setTraces\), we define the structure 
\(\overline{\setTraces}\) with domain \(\setTraces \cup \nat\) by letting,  for all \(a \in \Prop\), \(P_a = \{(\trace, k) \ | \ \trace \in \setTraces, k\in \nat \tAnd \trace[k](a) = 1 \}\) and 
\(\timeDef = \{(\trace, k) \ |\ \trace \in \setTraces, k \in \nat \tAnd 0 \leq k < |\trace| \}\).
A set of traces \(\setTraces\) is a model of a formula \(\varphi \in \TFOL\), denoted \(\setTraces \modelsTFOL \varphi\), when \(\overline{\setTraces}\) models \(\varphi\) under the classical first-order semantics.
From now on, we refer to \(P_a\) as \(a\) and omit the subscript \(\allSetTraces\) in \(\modelsTFOL\) whenever it is clear from the context.
The \emph{set of sets of traces generated by a hypertrace formula \(\varphi\)} is 
\(\generatedSet{\varphi} = \{\setTraces \ |\ \setTraces \modelsTFOL \varphi\}\).
We also equip Hypertrace Logic with a \emph{point interpretation} defined as \({\generatedSet{\varphi}_{\pointW} = \{\setTraces[0]\setTraces[1]\ldots \, |\, \setTraces \modelsTFOL \varphi\}}\),  where \({\setTraces[i] = \{\trace[i] \ |\ \trace \in \setTraces\}}\). 

\begin{example}
Consider the set of traces \({\setTraces = \{00\, (11)^\omega, 10\, (00)^\omega\}}\) with valuations over \((x,y)\).
Its point interpretation is \(\{00,10\}\,(\{11,00\})^\omega\).
\end{example}

\subsection{Trace-prefixed Hypertrace Logic}
\label{sub:tfol:hyperLTL}

\emph{Trace-prefixed Hypertrace Logic}, \(\TraceTFOL\), is a fragment of Hypertrace Logic in which all trace quantifiers are at the beginning of the formula. 
Its formulas, \(\varphi \in \TraceTFOL\), are defined by the following grammar: \(\varphi := \forall \traceVar\ \varphi \ |\ \neg \varphi\ |\ \psi\) with 
\(\psi := \forall \timeVar\ \psi \ |\ \psi \vee \psi \ |\ \neg \psi \ |\  \timeVar < \timeVar\ |\  \timeVar = \timeVar\ |\  a(\traceVar,\timeVar)\),
where 
\(\traceVar \in \Var\) is a trace variable, \(\timeVar\) is a time variable and \(a \in \Prop\) a propositional variable.

\(\TraceTFOL\) is expressively equivalent to HyperLTL~\cite{ClarksonFKMRS14} interpreted over sets of infinite traces. HyperLTL extends LTL by adding quantifiers over traces. Its syntax is defined by the following grammar, where \(\Var\) is a set of trace variables,
\(a \in \Prop\) and \(\pi  \in \Var\):
\(
\psi ::= \ \exists \pi\ \psi \ | \ \forall \pi \ \psi \ | \ \varphi\) with \(
\varphi ::= \ a_{\pi} \ |\ \neg \varphi \ |\ \varphi \vee \varphi \ | \ \Next \varphi \ | \ \varphi \Until \varphi.
\)
A trace assignment, \(\traceAssign_{\setTraces}: \mathcal{V} \rightarrow \setTraces\), is a partial function that assigns traces 
from \(\setTraces\) to trace variables in \(\Var\).
We denote by \(\traceAssign_{\setTraces}[\traceVar \mapsto \trace]\) the trace assignment in which \(\traceVar\) is mapped to \(\trace\) and otherwise identical to \(\traceAssign_{\setTraces}\).
The satisfaction relation for HyperLTL formulas is defined inductively as follows:
\begin{align*}
&(\traceAssign_{\setTraces}, i) \modelsHyper \exists \pi \ \psi
\tIff 
\text{there exists } \trace \in \setTraces: (\traceAssign_{\setTraces}[\pi \mapsto \trace], i) \modelsHyper \psi;\\
&(\traceAssign_{\setTraces}, i) \modelsHyper \forall \pi \ \psi 
\tIff 
\text{for all } \trace \in \setTraces: (\traceAssign_{\setTraces}[\pi \mapsto \trace], i) \modelsHyper \psi;\\
&(\traceAssign_{\setTraces}, i) \modelsHyper a_{\pi} 
\tIff  \  
\traceAssign_{\setTraces}(\pi)[i](a) = 1;\\
&(\traceAssign_{\setTraces}, i) \modelsHyper \neg  \psi 
\tIff \ 
(\traceAssign_{\setTraces}, i) \notmodelsHyper  \psi;\\
&(\traceAssign_{\setTraces}, i) \modelsHyper  \psi_1 \vee  \psi_2 \tIff \ 
(\traceAssign_{\setTraces}, i) \modelsHyper  \psi_1 \tOr (\traceAssign_{\setTraces}, i) \modelsHyper  \psi_2;\\
&(\traceAssign_{\setTraces}, i) \modelsHyper \Next  \psi \tIff \  (\traceAssign_{\setTraces}, i+1) \modelsHyper  \psi;\\
&(\traceAssign_{\setTraces}, i) \modelsHyper  \psi_1 \Until  \psi_2 \tIff 
\text{there exists } i \leq j:\! (\traceAssign_{\setTraces}, j) \modelsHyper  \psi_2  \text{ and  for all } i \leq j' < j:\! (\traceAssign_{\setTraces}, j') \modelsHyper  \psi_1.
\end{align*}
A set of traces \(\setTraces\) is a \emph{model} of a HyperLTL formula \(\varphi\), denoted by \(\setTraces \modelsHyper \varphi\), iff there exists a
mapping \(\traceAssign_{\setTraces}\)  s.t.\ \((\traceAssign_{\setTraces},0) \modelsHyper \varphi\).
A formula is closed when all occurrences of  trace variables are in the scope of a quantifier.
For all closed formulas (sentences)~\(\varphi\), \(\setTraces \modelsHyper \varphi\) iff \((\emptyAssign_{\setTraces},0) \modelsHyper \varphi\), 
where \(\emptyAssign_{\setTraces}\) is the empty assignment.
We may omit the subscript \(H\) in \(\models_H\) whenever it is clear from the context.

\begin{definition}
Let
\(\setTraces\) be a set of traces and 
\({\traceAssign_{\setTraces}: \mathcal{V} \rightarrow \setTraces}\) be a partial function assigning traces in 
\(\setTraces\) to variables in \(\mathcal{V}\).
We introduce the following notions:
\begin{itemize}
\item The set of \emph{trace variables assigned} in \(\traceAssign_{\setTraces}\) is \(\mathcal{V}(\traceAssign_{\setTraces}) = \{ \pi \ | \  \traceAssign_{\setTraces}(\pi) \text{ is defined} \}\);

\item This \emph{size} of \(\traceAssign_{\setTraces}\) is \(|\traceAssign_{\setTraces}| = |\mathcal{V}(\traceAssign_{\setTraces})|\);

\item The \emph{flattening} of \(\traceAssign_{\setTraces}\) is
\(\flatT{\traceAssign_{\setTraces}}[i](a_{\pi}) = \traceAssign_{\setTraces}(\pi)[i](a)\).
\end{itemize}
\end{definition}

Note that a quantifier-free HyperLTL formula $\varphi$ with trace variables in $\mathcal{V}$ and alphabet $X$ is also an LTL formula over the alphabet $\{ a_\pi \mid a \in X, \pi \in \mathcal{V}\}$.

\begin{example}
Let \(\setTraces = \{0^\omega, 1^\omega\}\) be a set of traces over \(\{a\}\). Consider, the assignment \(\traceAssign_{\setTraces}\) s.t.\ \(\traceAssign_{\setTraces}(\traceVar) = 0^\omega\) and \({\traceAssign_{\setTraces}(\traceVar') = 1^\omega}\).
Then, \(\flatT{\traceAssign_{\setTraces}}\) defines the trace \((01)^\omega\) over \((a_{\traceVar}, a_{\traceVar'})\).
\end{example}

\begin{proposition}\label{prop:zipping}
Let \(\varphi\) be a quantifier-free HyperLTL formula.
For all  \(i\in\nat\), all set of traces \(\setTraces\) and all of its trace assignments \(\traceAssign_{\setTraces}\),
\((\traceAssign_{\setTraces},i) \modelsHyper \varphi \tIff \flatT{\traceAssign_{\setTraces}}[i \ldots ] \models \varphi.\)
\end{proposition}

\begin{proposition}
For all HyperLTL sentences \(\varphi_H\) there exists a trace-prefixed hypertrace sentence \(\varphi\) s.t.\ for all sets of infinite traces \(\setTraces \subseteq \vals^{\omega}_{\Prop}\), \(\setTraces \modelsHyper \varphi_H\) iff
\(\setTraces \modelsTFOL \varphi\).
For all trace-prefixed hypertrace sentences \(\varphi\) there exists a HyperLTL sentence \(\varphi_H\) s.t.\ for all sets of infinite traces \(\setTraces \subseteq \vals^{\omega}_{\Prop}\), \(\setTraces \modelsHyper \varphi_H\) iff
\(\setTraces \modelsTFOL \varphi\).
\end{proposition}

\begin{proof}
The translation from HyperLTL formulas to an equivalent trace-prefixed hypertrace formula works as follows. 
We keep the trace quantifiers as they are and we use the translation from LTL to \(\FOLOrder\) introduced in~\cite{gabbay1980temporal} to translate the quantifier-free part. 
Then, we apply the following change in the quantifier-free part:
\(P_a(\traceVar, \timeVar) = P_{a_{\traceVar}}(\timeVar)\). 
Let us call this translation \(\text{tr}_{H}\). 
It follows from structural induction on HyperLTL formulas that
for all sets of traces and their assignments they satisfy an HyperLTL formula iff
they satisfy its translation to trace-prefixed hypertrace formula.
This follows from the result by Gabbay et al. in ~\cite{gabbay1980temporal} and
Proposition~\ref{prop:zipping} for the base case of this induction.
Hence
for all HyperLTL formulas \(\varphi_H\) there exists the trace-prefixed hypertrace formula
 \(\text{tr}_{H}(\varphi_H)\) s.t.\ \(\setTraces \modelsHyper \varphi_H\) iff \(\setTraces \modelsTFOL\text{tr}_{H}(\varphi_H)\). 

The translation from trace-prefixed hypertrace formulas to HyperLTL is similar. We use instead the translation from \(\FOLOrder\) to LTL from \cite{gabbay1980temporal}.
\end{proof}

\subsection{Time-prefixed Hypertrace Logic}

\emph{Time-prefixed Hypertrace Logic}, \(\TimeTFOL\), restricts the syntax of Hypertrace Logic to have all time constraints defined before trace quantifiers. 
Its formulas, \(\varphi \in \TimeTFOL\), are defined by the following grammar:
\(\varphi := \forall \timeVar\ \varphi \, |\, \neg \varphi\, |\, \timeVar < \timeVar\, |\,  \timeVar = \timeVar\, |\, \varphi \vee \varphi\ |\  \psi\) with \(
\psi := \forall \traceVar\ \psi \ |\ \psi \vee \psi \ |\ \neg \psi \ |\    a(\traceVar,\timeVar)\).
where 
\(\traceVar \in \Var\) is a trace variable, \(\timeVar\) is a time variable and \(a \in \Prop\) a propositional variable.


\emph{Globally Hypertrace logic}, \(\GFOL\), is a syntactic fragment of Time-prefixed Hypertrace logic
in which all formulas start with a universal time quantifier followed by a formula that can only have trace quantifiers. Then, 
\(\varphi \in \GFOL\) iff \(\varphi = \forall i\ \psi_i\) where \(\psi_i\) is defined by the following grammar:
\(\psi_{\timeVar} := \forall \traceVar\ \psi_{\timeVar} \ |\ \psi_{\timeVar} \vee \psi_{\timeVar} \ |\ \neg \psi_{\timeVar} \ |\   a(\traceVar,\timeVar) \ |\ \timeDef(\traceVar, \timeVar)\).

For a formula \(\psi(i)\) without time quantifiers and whose only free time variable is \(i\), we also define as a convenience its satisfaction w.r.t. sets of valuations \(\M = \{\val_0, \val_1, \ldots \}\) as follows:
\[\{\val_0, \val_1, \ldots \} \modelsTFOL \psi(i) \text{ iff 
for } \setTraces = \{\val_0^\omega, \val_1^\omega, \ldots\},\  \overline{\setTraces} \models \forall i\ \psi(i).\]

Globally Hypertrace Logic can be used to specify relations between traces of a system that must be satisfied in each of their time points independently. We use it later to specify the \emph{point semantics} of independence.

We prove below that if an hyperproperty can be expressed with globally hypertrace logic then it can be characterized by a set of sets of valuations \(\SetM\). 
We denote the set with all sequences of elements of \(\SetM\) by \(\SetM^{\omega}\).
In this formal language context, we interpret sets of valuations as letters.
Consider for instance the set \(\SetM = \{\{00, 01\}, \{11\}\}\) with valuations over \((x,y)\).
Then, \(\{00,11\}\,\{11\}\,\{00,11\}^\omega \in \SetM^{\omega}\) while \(\{00\}^{\omega} \notin \SetM^{\omega}\).

\begin{theorem}\label{thm:inv:global_FOL}
Let \(\Prop\) be a finite set of propositional variables and \(\setsetTraces \subseteq 2^{\vals^\omega_{\Prop}}\) be a hyperproperty.
If there exists a globally hypertrace formula \({\varphi \in \GFOL}\) that generates the same set of sets of traces as the hyperproperty,
\(\setsetTraces = \generatedSet{\varphi}\), then
there exists a set of sets of valuations \(\textbf{M} \subseteq 2^{\vals_{\Prop}}\) that generates the point-wise interpretation of the hyperproperty, 
\(\{ \setTraces[0] \setTraces[1] \ldots \ | \ \setTraces \in \setsetTraces  \} = \textbf{M}^{\omega}\), where \(\setTraces[i] = \{ \trace[i]\ |\ \trace \in \setTraces \}\).
\end{theorem}


\section[Two-state Local Independence]{Two-state Local Independence}
\label{sec:prop}

We are interested in specifying the following property:
\begin{center}
\emph{The value of \(y\) is independent from the value of \(x\) until 
$\state$ changes,
and from then on the value of \(z\) is independent from the value of \(x\).}
\end{center}


Independence requirements relate observable values from multiple system executions by requiring that for any pair of traces there exists a third that interleaves the first two. However, there is some freedom in how to combine and compare a pair of traces. 
In this work, we assume observations to be synchronous concerning the states of the specification. We can then compare observations either point-wise, with \emph{point semantics}, or as a whole, with \emph{segment semantics}. As independence requirements may be evaluated over sets with traces of different length, we compare a pair of traces with different size by matching their values up to the common length.
This enables us to capture dependencies between variables in systems where executions may stop at different points.
We could choose to compare only traces of the same size, this would not affect our results.

\begin{definition}
\label{def:independence}
Two variables, \(x\) and \(y\), are \emph{point independent}, denoted by
\(\ind{x}{y}{\syncS}\),~iff:
\[\forall i \forall \traceVar  \forall \traceVar'  \exists \traceVar_{\exists}\ 
\big(\timeDef(\traceVar, i)  \wedge \timeDef(\traceVar', i) \big) \rightarrow
\big(\timeDef(\traceVar_{\exists}, i) \wedge (x(\traceVar,i) \leftrightarrow x(\traceVar_{\exists},i)) \wedge 
(y(\traceVar',i) \leftrightarrow y(\traceVar_{\exists},i))\big).\]
\noindent Two variables, \(x\) and \(y\), are \emph{segment independent}, denoted by
\(\ind{x}{y}{\syncO}\), iff:
\[\forall \traceVar  \forall \traceVar'  \exists \traceVar_{\exists} \forall i \ 
\big(\timeDef(\traceVar, i)  \wedge \timeDef(\traceVar', i) \big) \rightarrow
\big(\timeDef(\traceVar_{\exists}, i) \wedge (x(\traceVar,i) \leftrightarrow x(\traceVar_{\exists},i)) \wedge 
(y(\traceVar',i) \leftrightarrow y(\traceVar_{\exists},i))\big).\]
\end{definition}

We need to distinguish between observations from the first and the second logical state of the observed system. For this, we define a slicing operator over sets of traces that returns all its elements prefixes (or suffixes) before (after) a given propositional variable holds for the first time.

\begin{definition}
\label{def:slicing}
Let \(a\) be a propositional variable.
The abbreviation, 
\(\text{min}(\traceVar, a, i)\), stands for
\(a(\traceVar, i) \wedge \forall j\ a(\traceVar, j) \rightarrow i \leq j.\)
Given a set of traces \(\setTraces\), we define its slicing w.r.t.\ \(a\), as follows:
\[
\setTraces[a\DDD]\! =\! \{ \trace[k\DDD] \, |\, \trace \in \setTraces, \ k \in \nat,  \text{min}(\trace, a, k)\} \tAnd 
\setTraces[\DDD a]\! =\! \{ \trace[\DDD k] \, |\, \trace \in \setTraces, \ k \in \nat,  \text{min}(\trace, a, k)\}.\]
\end{definition}

Remark that we only keep traces in which \(a\) holds at least once. The property that \(a\) holds at least once in every trace can be verified separately.

\begin{example}
Consider the set of traces \(\setTraces = \{00^\omega, 01 (10)^\omega\}\) in which the valuations are over \((a,x)\).
Then, \(\setTraces[\DDD a] = \{(00)^\omega, 01\}\) and \(\setTraces[a\DDD] = \{(10)^\omega\}\).
\end{example}

The action that triggers the change of state may occur at the same time point for all observations, be \emph{synchronous}, or at any time, be \emph{asynchronous}.

\begin{definition}
\label{def:prop:two_state_indep}
Let \(a\) be a boolean variable that is true when the state changes.
\emph{Two-state independence} is defined according to 
the possible action type, \(\sync\), \(\async\) or \(\novisible\), and w.r.t.\ an independence interpretation
\(\Ind \in \{ \Ind_{\syncS}, \Ind_{\syncO}\}\).
\begin{description}
\item [Asynchronous Action:]
\(\setsetTraces^{\async}_{\Ind} =  \{ \setTraces \, | \, \setTraces[\DDD a]\! \models \!\ind{x}{y}{} \tAnd  \setTraces[a \DDD]\! \models\! \ind{x}{z}{}\}.\)

\item [Synchronous Action:]
\(\setsetTraces^{\sync}_{\Ind} =  \{\setTraces \ | \  \setTraces\in \setsetTraces^{\async}_{\Ind} \tAnd
 \setTraces \models \exists i\, \forall \traceVar\ \text{min}(\traceVar,a,i)\}.\)

\item [Hidden Action:]
\(\setsetTraces^{\novisible}_{\Ind}\! \!=\! \{ \setTraces|_a \, | \, \exists a\, \setTraces[\DDD a]\! \models \!\ind{x}{y}{} \tAnd  \setTraces[a \DDD]\! \models\! \ind{x}{z}{}\},\)
where  
\(\setTraces|_{a}\) is the same set of traces as \(\setTraces\) except for the assignments of \(a\) that are removed. 
\end{description}
\end{definition}

We note that in the case that we cannot observe the action, we do not make 
any assumption on whether the actual underlying action is synchronous and do not impose any further restriction on it.

\section[Expressiveness]{Expressiveness}
\label{sec:exp}

In this section, we explore which variations of two-state local independence can be specified using Trace-prefixed Hypertrace Logic, which is expressible equivalent to HyperLTL. 
We summarize our results in Table~\ref{tab:exp:hyperLTL}.
\begin{table}
\caption{Trace-prefixed Hypertrace Logic expressiveness for different variants of the two-state local independence property. For the two-state independence under point semantics and synchronous action, we prove in Theorem \ref{thm:step} that it cannot be expressed by HyperLTL formulas with only a Globally operator.\label{tab:exp:hyperLTL}}
\centering
\begin{tabular}{c|c|c|c}
\textbf{Independence} & \multicolumn{3}{c}{\textbf{Action Timing}}\\
\cline{2-4}
\textbf{Semantics}& \emph{Sync}  & \emph{Async} & \emph{Hidden} \\ 
\hline
\emph{Point}  & No? [Thm. \ref{thm:step}] & No [Thm. \ref{thm:async:sync}] & No [Thm. \ref{thm:async:sync}]\\ 
\emph{Segment} & Yes [Thm. \ref{thm:sync:sync:overall}]  & No [Thm. \ref{thm:async:sync}] & No [Thm. \ref{thm:async:sync}]
\end{tabular} 
\end{table}


\subsection{Indistinguishable Trace Sets}
\label{sec:indist_trace_sets}

We introduce notions of indistinguishability between sets of traces for both the time-prefixed and the trace-prefixed fragments of Hypertrace Logic. 
We start by defining an equivalence between sets of traces for HyperLTL, which is expressively equivalent to the trace-prefixed fragment for sets with infinite traces. The number of trace quantifiers in a HyperLTL sentence defines how many traces can be compared in the requirement defined by the quantifier-free part. Recall that for quantifier-free formulas, HyperLTL satisfaction is reduced to LTL satisfaction, with assignments flattened to traces. We propose an equivalence notion for HyperLTL models that lifts equivalence between traces relative to a given class of LTL formulas to sets of traces. An example of such LTL equivalence is the stuttering equivalence between traces for the class of LTL formulas defined only with until modalities.

\begin{definition}
\label{def:equiv:ltl}
Let \(\LTLC\) be a class of LTL formulas.
We say that \(\Equiv_{\LTLC}\) is an \emph{equivalence on traces for formulas in \(\LTLC\)}
when \(\Equiv_{\LTLC}\) is an equivalence relation and for all LTL formulas \(\varphi \in \LTLC\) and
traces \(\trace \Equiv_{\LTLC} \trace'\), 
\(\trace \models \varphi \tIff  \trace' \models \varphi\).
\end{definition}

We extend classes of LTL formulas to classes of HyperLTL formulas based on their syntax, enabling us to characterize certain temporal aspects of HyperLTL.

\begin{definition}\label{def:hyperltl:classes}
Let \(\LTLC\) be a class of LTL formulas \(\LTLC\)
and let \(\varphi = Q_0 \traceVar_0 \ldots Q_n \traceVar_n \psi\) be a HyperLTL formula with \(\psi\) being quantifier-free and $Q_i \in \{\forall, \exists\}$, with \(0 \leq i \leq n\).
Then,  \(\varphi\) \emph{is in the HyperLTL extension of} \(\LTLC\), denoted \(\varphi \in \HyperLTLC\), iff \(\psi \in \LTLC\).
\end{definition}
Given an equivalence on traces for LTL formulas in a class \(\LTLC\), we extend it to HyperLTL formulas with \(k\) quantifiers followed by a temporal formula in \(\LTLC\)  by requiring a bijective translation between sets of traces that preserves \(\Equiv_{\LTLC}\), for all assignments of size \(k\).

\begin{definition}
\label{def:hyperLTL:global:equiv} 
Let \(k\in \nat\) and \(\LTLC\) be a class of LTL formulas, with \(\Equiv_{\LTLC}\) an equivalence on traces for formulas in \(\LTLC\).
Two sets of traces \(\setTraces\) and \(\setTraces'\) are \((k,\LTLC)\)-equivalent, denoted by
\(\setTraces \Equiv_{(k, \LTLC)} \setTraces'\), iff there exists a bijective and total function
\(\wit: \setTraces \rightarrow \setTraces'\), such that for all assignments over \(\setTraces\) and \(\setTraces'\)  of size \(k\), \(\traceAssign_{\setTraces}\) and \(\traceAssign_{\setTraces'}\), we have:
\({\flatT{\traceAssign_{\setTraces}} \Equiv_{\LTLC} \flatT{\wit(\traceAssign_{\setTraces})}}\) and
\(\flatT{\traceAssign_{\setTraces'}} \Equiv_{\LTLC} \flatT{\wit^{-1}(\traceAssign_{\setTraces'})}\).
We let \(\wit(\traceAssign)(\pi) = \wit(\traceAssign(\pi))\), if \(\pi \in \Var(\traceAssign)\) and, otherwise, undefined.
\end{definition}

\begin{theorem}\label{thm:exp:global}
Let \(\LTLC\) be a class of LTL formulas and \(\Equiv_{\LTLC}\) an equivalence on traces for formulas in \(\LTLC\).
For all HyperLTL sentences with quantifer-free part in the class \(\LTLC\), \(\varphi \in \HyperLTLC\), and for all two set of traces that are \((k, \LTLC)\)-equivalent, 
\(\setTraces \Equiv_{(|\Var(\varphi)|, \LTLC)} \setTraces'\), then, 
\(\setTraces \models \varphi\) iff \(\setTraces' \models \varphi\).
\end{theorem}
\begin{proof} Follows from the application of Lemma~\ref{lemma:exp:open} below.
\end{proof}

\begin{lemma}\label{lemma:exp:open}
Let \(\LTLC\) be a class of LTL formulas. For all HyperLTL formulas in the HyperLTL extension of  \(\LTLC\), \(\varphi \in \HyperLTLC\), and all sets of traces \(\setTraces\) and \(\setTraces\) that are \((|\Var(\varphi)|, \LTLC)\)-equivalent, \({\setTraces \Equiv_{(|\Var(\varphi)|, \LTLC)} \setTraces'}\), then for all functions \(\wit: \setTraces \rightarrow \setTraces'\) witnessing the equivalence and all assignments \(\traceAssign_{\setTraces}\) and 
\(\traceAssign_{\setTraces'}\) over \(\setTraces\) and \(\setTraces'\), respectively, only with assignments to the set of free variables in \(\varphi\),
\(\Var(\traceAssign_{\setTraces}) = \Var(\traceAssign_{\setTraces'}) = \fr(\varphi)\):
\[(\traceAssign_{\setTraces}, 0) \models \varphi \tIff (\wit(\traceAssign_{\setTraces}), 0) \models \varphi; \tAnd
(\traceAssign_{\setTraces'}, 0) \models \varphi \tIff (\wit^{-1}(\traceAssign_{\setTraces'}), 0) \models \varphi.\]
\end{lemma}
\begin{proof}
We proceed by structural induction on HyperLTL formulas in the extension of the LTL class \(\LTLC\).
The base case follows from \(\Equiv_{\LTLC}\) being an equivalence on traces for formulas in \(\LTLC\) and Proposition \ref{prop:zipping}.
We only treat the induction case for \(\forall \pi\ \varphi\) and  \(\exists \pi\ \varphi\), the full proof is in appendix.

Assume by induction hypothesis (IH) that the statement holds for arbitrary \(\varphi \in \HyperLTLC\). 
Assume that \emph{(i)}
\(\setTraces \Equiv_{(|\Var(\forall \pi\ \varphi)|, \LTLC)} \setTraces'\).
Note that, wlog we can assume that quantifiers bind a variable already occurring in \(\varphi\), i.e.\
 \(|\Var(\forall \pi\ \varphi)| = |\Var(\varphi)|\). Then, \emph{(i')} \({\setTraces \Equiv_{(|\Var(\varphi)|, \LTLC)} \setTraces'}\), and it has the same witnesses as assumption~\emph{(i)}. 
Let \(\wit: \setTraces \rightarrow \setTraces'\) be a function that witnesses \emph{(i)}.
Now, consider arbitrary \(\traceAssign_{\setTraces}\) and \(\traceAssign_{\setTraces'}\), over \(\setTraces\) and \(\setTraces'\), s.t.\ \(\Var(\traceAssign_{\setTraces}) = \Var(\traceAssign_{\setTraces'}) =\fr(\forall \pi\ \varphi) = \fr(\varphi) \setminus\{\pi\}\). 
We prove next that, if  \((\traceAssign_{\setTraces}, 0) \models \forall \pi \ \varphi\) then \({(\wit(\traceAssign_{\setTraces}), 0) \models \forall \pi\ \varphi}\).

Assume that \({(\traceAssign_{\setTraces}, 0) \models \forall \pi \ \varphi}\), then
\((\star)\ \text{for all } \trace \in \setTraces: (\traceAssign_{\setTraces}[\pi \mapsto \trace], 0) \models \varphi.\)
By Definition~\ref{def:hyperLTL:global:equiv}, \(\Var(\wit(\traceAssign_{\setTraces})) = \Var(\traceAssign_{\setTraces})\). Thus, \(\Var(\wit(\traceAssign_{\setTraces})[\pi \mapsto \trace']) =
\Var(\traceAssign_{\setTraces}) \cup \{\pi\} = \fr(\varphi)\). We can apply the (IH), because \({\setTraces \Equiv_{(|\Var(\varphi)|, \LTLC)} \setTraces'}\), \(\wit\) witnesses it,
and for all \(\trace \in \setTraces\) then \(\traceAssign_{\setTraces}[\pi \mapsto \trace]\) is an assignment over \(\setTraces\). So, 
\(\text{for all } \trace \in \setTraces: (\wit(\traceAssign_{\setTraces}[\pi \mapsto \trace]), 0) \models \varphi.\)

Assume towards a contradiction that \((\wit(\traceAssign_{\setTraces}), 0) \not \models \forall \pi\ \varphi\).
Then,
\(\text{there exists } \trace' \in \setTraces'\) s.t.\ \((\wit(\traceAssign_{\setTraces})[\pi \mapsto \trace'], 0) \not \models \varphi.\)
We can apply the (IH), because
\(\Var(\wit(\traceAssign_{\setTraces})[\pi \mapsto \trace']) =  \fr(\varphi)\),  \emph{(i')} with \(\wit\) being one of its witnesses,
and for all \(\trace' \in \setTraces'\) then \(\wit(\traceAssign_{\setTraces})[\pi \mapsto \trace']\) is an assignment over \(\setTraces'\). Hence there exists \(\trace' \in \setTraces'\) s.t.\
\({(\wit^{-1}(\wit(\traceAssign_{\setTraces})[\pi \mapsto \trace']), 0) \not \models \varphi}.\)
Then, by  Definition~\ref{def:hyperLTL:global:equiv}, there exists
\(\trace' \in \setTraces'\) s.t.\ 
\((\wit^{-1}(\wit(\traceAssign_{\setTraces}))[\pi \mapsto \wit^{-1}(\trace')], 0) \not \models \varphi.\)
As \(\wit\) is a bijective function, \((\wit^{-1}(\wit(\traceAssign_{\setTraces})) = \traceAssign_{\setTraces}\), and so there exists
\(\trace' \in \setTraces'\) s.t.\ \((\traceAssign_{\setTraces}[\pi \mapsto \wit^{-1}(\trace')], 0) \not\models \varphi.\)
And this is equivalent to, there exists \(\trace' \in \setTraces'\) s.t.\
\(\trace = \wit^{-1}(\trace') \tAnd (\traceAssign_{\setTraces}[\pi \mapsto\trace], 0) \not\models \varphi.\)
Given that \(\wit\) is a surjective function, then
\(\text{there exists } \trace \in \setTraces:  (\traceAssign_{\setTraces}[\pi \mapsto\trace], i) \not\models \varphi.\)
This contradicts \((\star)\).

We prove now that if \((\traceAssign_{\setTraces}, 0) \models \exists \pi \ \varphi\) then \({(\wit(\traceAssign_{\setTraces}), 0) \models \exists \pi\ \varphi}\).
Assume that \((\traceAssign_{\setTraces}, 0) \models \exists \pi \ \varphi\), then there exists \(\trace \in \setTraces\) s.t.\
\((\traceAssign_{\setTraces}[\pi \mapsto \trace], 0) \models \varphi.\)
By Definition~\ref{def:hyperLTL:global:equiv}, \(\Var(\wit(\traceAssign_{\setTraces})) = \Var(\traceAssign_{\setTraces})\), and thus \(\Var(\wit(\traceAssign_{\setTraces})[\pi \mapsto \trace']) =
\Var(\traceAssign_{\setTraces}) \cup \{\pi\} = \fr(\varphi)\). Then, we can apply the (IH), because \(\setTraces \Equiv_{(|\Var(\varphi)|, \LTLC)} \setTraces'\) with \(\wit\) being one of its witnesses, and  for all \(\trace \in \setTraces\) then \(\traceAssign_{\setTraces}[\pi \mapsto \trace]\) is an assignment over \(\setTraces\). So, there exists \(\trace \in \setTraces\) s.t.\
\((\wit(\traceAssign_{\setTraces}[\pi \mapsto \trace]), 0) \models \varphi.\)
By Definition~\ref{def:hyperLTL:global:equiv} and \(\wit\) being a total function, then
\(\text{there exists } \trace \in \setTraces, \trace'\! =\!\wit(\trace) \tAnd (\wit(\traceAssign_{\setTraces})[\pi \mapsto \trace'], 0)\! \models\! \varphi.\)
Hence
\(\text{there exists } \trace' \in \setTraces': (\wit(\traceAssign_{\setTraces})[\pi \mapsto \trace'], 0) \models \varphi\), and so
\((\wit(\traceAssign_{\setTraces}), 0) \models \exists \pi\ \varphi\).
\end{proof}

\begin{remark}
The other direction of the implication in Theorem~\ref{thm:exp:global} does not hold.
Consider the two set of traces below with valuations over \((x)\): 
\(\setTraces\! =\! \{1010^\omega\}\) and \({\setTraces\! =\! \{10010^\omega, 100010^\omega\}}\).
The sets have different cardinally, so there is no \(k\) and \(\LTLC\) for each they are \((k, \LTLC)\)-equivalent. 
However they are indistinguishable for all HyperLTL formulas with one trace quantifier and only until modalities, because the traces in \(\setTraces'\) are stutter-equivalent to the trace in \(\setTraces\).
\end{remark}


Next we introduce some notions of equivalence over traces that are used later in our results.
We start by defining \emph{Globally LTL}, \(\LTLBox\), a LTL class with all formulas that have Globally, \(\Globally\), as the topmost and unique modal operator.
Then, 
\({\LTLBox = \{ \Globally \psi \ |\ \psi \text{ is a propositional formula}\}}\).

\begin{definition}\label{def:globally:equiv}
Two traces \(\trace\) and \(\trace'\) are \(\Equiv_{\LTLBox}\) equivalent iff
\(\{\trace[i] \ | \ i \in \nat\} = \{\trace'[j] \ | \ j \in \nat\}\).
\end{definition}

\begin{proposition}\label{prop:ltlbox:equiv}
For all all two traces \(\trace\)
and~\(\trace'\),
\(\trace \Equiv_{\LTLBox} \trace'\) iff, for all \(\varphi \in \LTLBox\),
\(\trace \models \varphi\) iff \(\trace' \models \varphi\).
\end{proposition}

The class \(\LTLX^n\) is the class of LTL formulas with up to \(n\) nesting of the \(\Next\) operator. In~\cite{kuvcera2005stuttering}, the authors introduce the notion of a letter being \(n\)-redundant in a trace. A letter is \(n\)-redundant if it is repeated for at least \(n\) consecutive times. Then, two traces are \(n\)-stutter equivalent if they are equal up to the deletion of \(n\)-redundant letters. The following Proposition~\ref{prop:redundant} is a direct consequence of the results in  \cite{kuvcera2005stuttering}.

\begin{definition}[\cite{kuvcera2005stuttering}]
A valuation at a time point \(i\) in a trace \(\trace\) is \(n\)-redundant iff 
\(\trace[i] = \trace[i+j]\) for all \(1 \leq j \leq n+1\).
For all \(n \in \nat\) we define the relation \(\prec_{n}\) over finite or infinite traces, as follows: \(\trace \prec_n \trace'\) iff \(\trace\) can be obtained from \(\trace'\) by deleting a \(n\)-redundant valuation. The relation \(\Equiv_n\) is the least equivalence over the set of all finite or infinite traces containing \(\prec_n\). Then, the traces \(\trace\) and \(\trace'\) are \(n\)-stutter equivalent iff \(\trace \Equiv_n \trace'\).
\end{definition}
\begin{proposition}[\cite{kuvcera2005stuttering}]\label{prop:redundant}
For all formulas \(\varphi \in \LTLX^n\), if  \(\trace \Equiv_n \trace'\), then \(\trace \models \varphi\) iff \(\trace' \models \varphi\).
\end{proposition}

We introduce a notion of indistinguishable sets of trace for time-prefixed Hypertrace logic.
Consider a time-prefixed formula that quantifies over \(k\) time points. Then, two sets of traces are \(k\)-point equivalent if for each possible \(k\)-tuple of time points there is a bijective translation between the sets of traces that makes them indistinguishable in the times of that tuple.

\begin{definition}
Two sets of traces, \(\setTraces\) and \(\setTraces'\), are \emph{\(k\)-point equivalent}, denoted by \({\setTraces \Equiv^{\syncS}_k \setTraces'}\), iff
for all \(k\)-tuples of time positions, \((i_1, \DDD i_k) \in \nat^k\), there exists a bijective and total function \(\wit: \setTraces \rightarrow \setTraces'\) s.t.\ for all \(\trace \in \setTraces\) we have 
\(\trace[i_j] = \wit(\trace)[i_j]\) and for all \(\trace' \in \setTraces'\) we have 
\(\trace'[i_j] = \wit^{-1}(\trace)[i_j]\), with \(1 \leq j \leq k\).
\end{definition}

\begin{theorem}\label{thm:time:indist}
For all time-prefixed Hypertrace sentences \(\varphi \in \TimeTFOL\) and all sets of traces, \(\setTraces\) and \(\setTraces'\), that are \(|\timeVars(\varphi)|\)-point equivalent, \(\setTraces \Equiv^{\syncS}_{|\timeVars(\varphi)|} \setTraces'\), where  \(|\timeVars(\varphi)|\) is the number of time variables in \(\varphi\), then 
\(\setTraces \models \varphi\) iff \(\setTraces' \models \varphi\).
\end{theorem}

\subsection{Point Semantics}
\label{sub:step_semantics}

The point semantics interpretation of independence considers each time point independently. Recall from Definition \ref{def:independence} that \(\ind{x}{y}{\syncS}\) is defined as:
\[\forall i \forall \traceVar  \forall \traceVar'  \exists \traceVar_{\exists}\ 
\big(\timeDef(\traceVar, i)  \wedge \timeDef(\traceVar', i) \big) \rightarrow
\big(\timeDef(\traceVar_{\exists}, i) \wedge (x(\traceVar,i) \leftrightarrow x(\traceVar_{\exists},i)) \wedge 
(y(\traceVar',i) \leftrightarrow y(\traceVar_{\exists},i))\big).\]

Globally HyperLTL is the extension of Globally LTL with trace quantifiers. 
We start by proving that no Globally HyperLTL formula can express  \emph{one-state independence with point semantics}, 
\(\setsetTraces^1_{\syncS} = \generatedSet{\ind{x}{y}{\syncS}}\).
Note that \(\ind{x}{y}{\syncS}\) is a Globally Hypertrace formula.

First, we define two families of models parameterized by a natural number s.t.\ one of them satisfies the one-state independence with point semantics while the other does not satisfy it. The parameter in the models guarantees that given a HyperLTL with \(n\) quantifiers there are enough traces in the models to prevent HyperLTL from distinguishing them.
We exploit the fact that while evaluating a HyperLTL formula we can compare simultaneously at most the same number of traces as the number of quantifiers.
Then, we prove that no Globally HyperLTL formula can distinguish between the two types of models. 
To prove this result, we show that there exists a \((k, \LTLBox)\)-equivalence between the models, where \(\Equiv_{\LTLBox}\) is an equivalence over traces for Globally LTL formulas.

\begin{definition}\label{def:models:step:sync_action}
We define below set of sets of traces \(\setTraces_n^{\pointW}\) and
\({\setTraces'}_n^{\pointW}\), for \(n\in \nat\) and with valuations over \((x,y)\):
\begin{align*}
&E_n = \{(11)^{n+2}(00)^\omega\} \cup \bigcup\limits_{0\leq j < n}\{(00)^{j}\ 10\  (00)^\omega,
(00)^{j}\ 01\  (00)^\omega\};\\
&\setTraces^{\pointW}_{n}= E_n \cup \{(00)^{n}\ 10\ 10\  (00)^\omega, (00)^{n}\ 01\ 01\  (00)^\omega\}; \tAnd\\
&{\setTraces'}^{\pointW}_{n} =E_n \cup \{(00)^{n}\ 10\ 00\  (00)^\omega, (00)^{n}\ 01\ 00\  (00)^\omega\}.
\end{align*}
\end{definition}

\begin{example}
For \(n=1\), we get the following sets of traces:
\begin{align*}
\setTraces^{\pointW}_1 = \{& 11\ 11\ 11\  (00)^\omega, 
&\setTraces'^{\pointW}_1 = \{ & 11\ 11\ 11\  (00)^\omega,\ \\
& 10\ 00\ 00\  (00)^\omega, && 10\ 00\ 00\  (00)^\omega,\ \\
& 01\ 00\ 00\  (00)^\omega, && 01\ 00\  00\  (00)^\omega,\ \\
& 00\ 10\  10\  (00)^\omega, && 00\ 10\ 00\  (00)^\omega,\ \\
& 00\ 01\ 01\  (00)^\omega\} && 00\ 01\ 00\   (00)^\omega\}
\end{align*}
The set \(\setTraces^{\pointW}_1\) satisfies the condition that \(x\) is independent of \(y\) because at all time points we have all possible combinations of observations for \(x\) and \(y\). However, \(\setTraces^{\pointW}_1\) does not satisfy the requirement, because at time 2 we are missing traces with valuations \(10\) and \(01\) in \((x,y)\). Globally HyperLTL formulas with only one trace quantifier cannot distinguish between these sets of traces. 
\end{example}

\begin{lemma}\label{lemma:step:settraces_and_prop}
\(\setTraces^{\pointW}_{n} \in  \setsetTraces^1_{\syncS}\) and 
\({\setTraces'}^{\pointW}_{n} \notin  \setsetTraces^1_{\syncS}\). 
\end{lemma}

\begin{lemma} \label{lemma:step:settraces:equiv} For all \(n\in \nat\),
\(\setTraces^{\pointW}_{n} \Equiv_{(n, \LTLBox)} {\setTraces'}^{\pointW}_{n}\).
\end{lemma}

\begin{theorem}\label{thm:step}
For all Globally HyperLTL formulas \(\varphi \in 2^{\LTLBox}\),
\(\generatedSet{\varphi} \neq \setsetTraces^{\sync}_{\syncS}\).
\end{theorem}
\begin{proof}
By  Lemma \ref{lemma:step:settraces_and_prop}, 
Lemma \ref{lemma:step:settraces:equiv} and Theorem \ref{thm:exp:global}, it follows that 
for all HyperLTL formulas in the class extending globally LTL, \(\varphi \in 2^{\LTLBox}\), 
\(\generatedSet{\varphi_{\LTLBox}} \neq \setsetTraces^{1}_{\syncS}\).
Assume towards a contradiction that there exists a Globally HyperLTL formula \(\varphi\) s.t.\ 
\(\generatedSet{\varphi} = \setsetTraces^{\sync}_{\syncS}\).
Then, we define \(\varphi_y = \varphi[z \mapsto y]\) where \([z \mapsto y]\) substitutes all occurrence of \(z\) by \(y\).
Then, \(\generatedSet{\varphi_y} = \setsetTraces^{1}_{\syncS}\).
This is a contradiction, and so for all Globally HyperLTL formulas \(\varphi\), \(\generatedSet{\varphi} \neq \setsetTraces^{\sync}_{\syncS}\).
\end{proof}

We conjecture that this result extends to all HyperLTL formulas. 
Globally hypertrace formulas enforce a requirement over all time points that must be satisfied independently by them. Intuitively, such properties can be only expressed with HyperLTL formulas that are equivalent to a globally HyperLTL formula.

It is not surprising that time-prefixed hypertrace formulas can express two-state
independence under point semantics with synchronous action. We conjecture that this is the only variant it can express.

\begin{theorem}\label{thm:time:point:sync}
Consider the following time-prefixed hypertrace formula:
\[
\begin{split}
\varphi^{\sync}_{\text{time}} \overset{\text{def}}{=} & \ \exists j \forall i < j \forall k \leq j \forall \traceVar \forall \traceVar' \exists \traceVar_{\exists} \\
&\ \ \big(
\neg a(\traceVar, i) \wedge \neg a(\traceVar', i)  \wedge
(x(\traceVar,i) \leftrightarrow x(\traceVar_{\exists},i)) \wedge 
(y(\traceVar',i) \leftrightarrow y(\traceVar_{\exists},i))\big) \wedge \\
&\ \ \big(
a(\traceVar, j) \wedge a(\traceVar', j)  \wedge 
(x(\traceVar,k) \leftrightarrow x(\traceVar'_{\exists},k)) \wedge 
(z(\traceVar',k) \leftrightarrow z(\traceVar'_{\exists},k))\big)
\end{split}
\]
Then, \(\generatedSet{\varphi^{\sync}_{\text{time}}} =   \setsetTraces^{\sync}_{\syncS}\).
\end{theorem}

\subsection{Segment Semantics}
\label{sub:segment_semantics}

The segment semantics of independence compares between whole observations of a state in a system.
Recall from Definition \ref{def:independence} that \(\ind{x}{y}{\syncO}\) is defined as:
\[\forall \traceVar  \forall \traceVar'  \exists \traceVar_{\exists} \forall i \ 
\big(\timeDef(\traceVar, i)  \wedge \timeDef(\traceVar', i) \big) \rightarrow
\big(\timeDef(\traceVar_{\exists}, i) \wedge (x(\traceVar,i) \leftrightarrow x(\traceVar_{\exists},i)) \wedge 
(y(\traceVar',i) \leftrightarrow y(\traceVar_{\exists},i))\big).\]

We prove that HyperLTL can express the two-state segments independence with synchronous action, while both asynchronous and hidden action are not expressible.

The intuitive HyperLTL formula for the two-state segments independence entails that the action is synchronous.  So, we already cannot expect to rely on the proposition \(a\) to slice our traces accurately, when the action is asynchronous. To prove that HyperLTL cannot express the property in this scenario, we exploit the fact that we need to compare arbitrarily distant time points from different observations. 

\begin{theorem}\label{thm:sync:sync:overall}
Consider the following HyperLTL formula: 
\[\varphi^{\sync}_{\syncO} \overset{\text{def}}{=}
\forall \traceVar \forall \traceVar' \exists \traceVar_{\exists} \exists \traceVar'_{\exists}\, 
(\neg a_{\traceVar} \wedge \neg a_{\traceVar'} \wedge x_{\traceVar} = x_{\traceVar{\exists}} \wedge y_{\traceVar'} = y_{\traceVar_{\exists}}\!) \Until
(a_{\traceVar} \wedge a_{\traceVar'} \wedge \Box (x_{\traceVar} = x_{\traceVar_{\exists}} \wedge z_{\traceVar'} = z_{\traceVar'_{\exists}}))\]
Then, \(\generatedSet{\varphi^{\sync}_{\syncO}} =   \setsetTraces^{\sync}_{\syncO}\).
\end{theorem}

We now examine the case of an asynchronous action.
Like in the previous section for point semantics,
we start by defining a family of models s.t.\
one of the families satisfies the two-state independence property while the other does not. The difficulty in expressing the asynchronous action is the arbitrary distance between time points we want to compare. Thus, we create the models to guarantee that there are not enough \emph{next} operators to encode this distance.
Then, the second family is the same as the first except for the position \(2n+1\) that is deleted. This position will coincide with a global (across all sets in the set of traces) \(n\)-stuttering in the first family. Thus, it is not surprising that instances of these families, for a given \(n\in \nat\), are \((k,\LTLX^n)\)-equivalent, for any number of trace quantifiers \(k\).

\begin{definition}\label{def:models:async_action}
The sets of sets of traces \(\setTraces_n^{\async}= \{t_1, t_2, t_3, t_4\}\) and
\({\setTraces'}_n^{\async} = \{t_1', t_2', t_3', t_4'\}\), for \(n\in \nat\), with valuations over \((a,x,y,z)\) are defined by letting
\begin{align*}
\trace_0 &\!=\!\! 1110\, (1000)^{n+4}\, (1001)^{n+4}\, 1111\, (1001)^{n+4}\, (1000)^{n+4},\\
\trace_1 &\!=\!\! 1111\, (1001)^{n+4}\, (1000)^{n+4}\, 1110\, (1000)^{n+4}\, (1001)^{n+4},\\
t_1 &\!=\! 0000\ \trace_1\ (1001)^{\omega},\  t_2 = 0010 \ \trace_1\ (1001)^{n+4}\ (1111)^{\omega},\\ 
t_3 &\!=\! (0000)^{n+4}\ \trace_0\  (1001)^{\omega},\ t_4 = (0010)^{n+4}\ \trace_0\  (1111)^{\omega},\\
t_i' &\!=\! t_i[0]t_i[1]\ldots t_i[2n+10] t_i[2n+12] \ldots \text{ for } 1\leq i \leq 4.
\end{align*}
\end{definition}

\begin{lemma}\label{lemma:nredundant:traces_aync}
For all assignments \(\traceAssign_{\setTraces_n^{\async}}\) over \(\setTraces_n^{\async}\), the valuation at \(2n+11\) is \(n\)-redundant in the trace \(\flatT{\traceAssign_{\setTraces_n^{\async}}}\).
\end{lemma}

It is clear, that all sets of traces that are models under the segments semantics are models under the point semantics, as well. Then, 
\(\setsetTraces^{\async}_{\syncO} \subseteq  \setsetTraces^{\async}_{\syncS}\).

\begin{lemma}\label{lemma:async:sync}
\(\setTraces_n^{\async} \in \setsetTraces^{\async}_{\syncO}\), \({\setTraces'}_n^{\async} \not\in \setsetTraces^{\async}_{\syncS}\) and 
\({\setTraces'}^{\async}_n|_{a} \not\in \setsetTraces^{\novisible}_{\syncS}\).
\end{lemma}

We remark that the set of traces \(\setTraces_n^{\async}\) satisfies the two-state independence even when the segment interpretation of independence compares only pairs of traces of the same length.

\begin{lemma}\label{lemma:models:hyperltl:indist}
For all \(n\in \nat\), \(k \in \nat\) and HyperLTL formulas \(\varphi \in \LTLX^n\), 
\(\setTraces_n^{\async} \Equiv_{(k,\LTLX^n)} {\setTraces'}_n^{\async}\) and 
\(\setTraces_n^{\async}|_a \Equiv_{(k,\LTLX^n)} {\setTraces'}_n^{\async}|_a\).
\end{lemma}

\begin{proof}
Consider arbitrary \(n\in \nat\) and \(k \in \nat\). 
We define the witness function \(\wit: \setTraces_n^{\async} \rightarrow {\setTraces'}_n^{\async}\) as \(\wit(t_i) = t_i'\), with \(1 \leq i \leq 4\). Clearly, it is both bijective and total.
Let \(\traceAssign_{\setTraces_n^{\async}}\) be an arbitrary assignment over \(\setTraces_n^{\async}\) s.t.\ 
\(|\traceAssign_{\setTraces_n^{\async}}| = k\). We prove in Lemma \ref{lemma:nredundant:traces_aync} that 
the letter at \(2n+11\) in \(\flatT{\traceAssign_{\setTraces_n^{\async}}}\) is \(n\)-redundant. By definition of \({\setTraces'}_n^{\async}\) , 
\(\flatT{\wit(\traceAssign_{\setTraces_n^{\async}})}\) is the same as \(\flatT{\traceAssign_{\setTraces_n^{\async}}}\) except for the valuation at \(2n+11\) that is deleted. Then, 
\(\flatT{\traceAssign_{\setTraces_n^{\async}}} \Equiv_{\LTLX^n} \flatT{\wit(\traceAssign_{\setTraces_n^{\async}})}\).
We prove analogously that for all assignments over \({\setTraces'}_n^{\async}\), \(\traceAssign_{{\setTraces'}_n^{\async}}\), with size \(k\),
\(\flatT{\traceAssign_{{\setTraces'}_n^{\async}}} \Equiv_{\LTLX^n} \flatT{\wit^{-1}(\traceAssign_{{\setTraces'}_n^{\async}})}\).
Hence \(\setTraces_n^{\async} \Equiv_{(k,\LTLX^n)} {\setTraces'}_n^{\async}\).
We use the same witness function to prove that 
\(\setTraces_n^{\async}|_a \Equiv_{(k,\LTLX^n)} {\setTraces'}_n^{\async}|_a\).
Note, as \(\setTraces_n^{\async}|_a\) is the same as \(\setTraces_n^{\async}\) except for the valuations of \(a\) that are removed,  then Lemma \ref{lemma:nredundant:traces_aync} holds for \(\setTraces_n^{\async}|_a\), as well.
\end{proof}

\begin{theorem}\label{thm:async:sync}
For all HyperLTL sentences \(\varphi\): 
\(\generatedSet{\varphi} \neq \setsetTraces^{\async}_{\syncS}\), 
\(\generatedSet{\varphi} \neq \setsetTraces^{\async}_{\syncO}\), 
\(\generatedSet{\varphi} \neq \setsetTraces^{\novisible}_{\syncS}\) and
\(\generatedSet{\varphi} \neq \setsetTraces^{\novisible}_{\syncO}\).
\end{theorem}
\begin{proof}
From \(\setsetTraces^{\async}_{\syncO} \subseteq  \setsetTraces^{\async}_{\syncS}\) and Lemma~\ref{lemma:async:sync}, it follows that:
\begin{itemize}
\item \(\setTraces_n^{\async} \in \setsetTraces^{\async}_{\syncO}\) and
 \({\setTraces'}_n^{\async} \not\in \setsetTraces^{\async}_{\syncO}\); and
\item \(\setTraces_n^{\async} \in \setsetTraces^{\async}_{\syncS}\) and \({\setTraces'}_n^{\async} \not\in \setsetTraces^{\async}_{\syncS}\).
\end{itemize}

Let \(\varphi\) be a closed HyperLTL formulas and let \(n\) be the number of its nested next operators.
Then, \(\varphi \in 2^{\LTLX^n} \) and there exists \(k \in \nat\) equal to the number of variables in \(\varphi\).
So, it follows from Lemma~\ref{lemma:models:hyperltl:indist} and 
 Theorem~\ref{thm:exp:global} that \(\setTraces_n^{\async} \in \generatedSet{\varphi}\)
 iff \({\setTraces'}_n^{\async} \in \generatedSet{\varphi}\).
Hence  for all HyperLTL sentences \(\generatedSet{\varphi} \neq \setsetTraces^{\async}_{\syncS}\) and
\({\generatedSet{\varphi} \neq \setsetTraces^{\async}_{\syncO}}\).

As  \(\setTraces_n^{\async}|_{a}\) is the same as \(\setTraces_n^{\async}\), except for the valuations of \(a\) that were removed, then
 \(\setTraces_n^{\async}|_{a} \in \setsetTraces^{\novisible}_{\syncO}\) and
\(\setTraces_n^{\async}|_{a} \in \setsetTraces^{\novisible}_{\syncS}\).
By Lemma~\ref{lemma:async:sync}, 
 \({\setTraces'}^{\async}_n|_{a} \not\in \setsetTraces^{\novisible}_{\syncS}\) and so it follows that \({\setTraces'}^{\async}_n|_{a} \not\in \setsetTraces^{\novisible}_{\syncO}\).
 As in the previous case, from Lemma~\ref{lemma:models:hyperltl:indist} and 
 Theorem~\ref{thm:exp:global}, it follows  that 
 for all HyperLTL formulas \(\generatedSet{\varphi} \neq \setsetTraces^{\novisible}_{\syncS}\) and
\(\generatedSet{\varphi} \neq \setsetTraces^{\novisible}_{\syncO}\).
\end{proof}


\section{Related Work}
\label{sec:related_work}

Linear-time hyperlogics support the comparison between traces from a given set.
Trace properties, often specified in LTL~\cite{Pnueli77}, 
are not expressive enough to specify such relations \cite{mclean1996general, ClarksonS10}. 
The seminal work of Clarkson and Schneider~\cite{ClarksonS10} introduces the concept of  \emph{hyperproperties} as
sets of trace properties.

Different extensions to LTL have been proposed to reason about security properties
that often require comparing multiple executions of a system.  
Well-known examples are the epistemic temporal logic (ETL)~\cite{fagin1995reasoning}, which extends LTL with the \emph{epistemic modal operator for knowledge}; and SecLTL \cite{dimitrova2012model}, which introduces the \emph{hide} modality.
As an attempt to define a unifying logic for hyperproperties, Clarkson et al. introduce HyperLTL \cite{ClarksonFKMRS14}, which extends LTL with explicit quantification over traces.

The hide operator in SecLTL
considers all alternative outcomes from the current time. For this reason, in \cite{ClarksonFKMRS14} the authors argue that there is a SecLTL formula that can distinguish between some systems with different computations paths but the same set of traces.
In the same paper, they prove that HyperLTL subsumes ETL. 
Their proof relies on the possibility to quantify over propositional variables that are not part of the system that generates a given set of traces. Later they updated the definition of HyperLTL to not allow such quantification. 
This extension to HyperLTL, with quantification over propositional variables, is introduced in \cite{coenen2019hierarchy}  as HyperQPTL 
 and proven to be strictly more expressive than HyperLTL. 

Bozzelli et al. prove, in \cite{bozzelli2015unifying},  that CTL* extended with trace quantifiers (HyperCTL*) and with the knowledge operator
(KCTL*) have incomparable expressive power. These results extend to HyperLTL and ETL, as well, as they are both subsumed by the respective CTL* extension. 
They start by proving that no ETL formula can specify that in a given set of traces two traces only differ at a time point, which can be specified in HyperLTL. Their result explores the fact that
trace quantification in ETL is implicit, as the only way to compare different traces is with the knowledge operator. 
Later, they prove that HyperLTL cannot express 
bounded termination.
This result relies on the fact that, for all HyperLTL formulas, time quantifiers are always dependent on the trace quantifiers. The latter property can be specified in ETL.

In contrast to the extensions to LTL discussed above, in \cite{krebs2018team},  Krebs et al. propose to reinterpret LTL under \emph{team semantics}. Team semantics works with sets of assignments, referred to as \emph{teams}. The authors introduce synchronous and asynchronous semantics. Similar to how we specify two-state independence, their semantics differ on how they slice the set of traces while interpreting the time operators. Synchronous semantics requires the time to be global, while in the asynchronous case time is local to each trace. 
They show that HyperLTL and LTL under team semantics and synchronous entailment have incomparable expressive power.

Previous negative expressivity results about HyperLTL in the literature refer to the property used in the proof by Bozelli et al. in \cite{bozzelli2015unifying}. Their proof defines an equivalence relation for a specific family of models to show that no HyperCTL* can distinguish them. 
To the best of our knowledge, only 
Finkbeiner and Rabe \cite{finkbeiner2014linear} identify an equivalence relation over sets of traces that are not distinguishable by HyperLTL formulas.  
Similar to LTL, HyperLTL cannot distinguish between systems that generate the same set of traces.

In \cite{finkbeiner2017first}, the authors propose to extend \(FO[<]\) to hyperproperties by adding the equal-level predicate \(E\) and denote this extension as \(FO[<,E]\). 
In this approach, time positions are labeled by traces and the predicate \(E\) is intended to relate the same time positions occurring in different traces. They prove in the same paper that \(FO[<,E]\) is more expressive than HyperLTL. They then define HyperFO by distinguishing quantifiers over initial positions (equivalent to trace quantifiers in HyperLTL) from time quantifiers (ordinary temporal operators in HyperLTL). Finally, they prove that HyperFO and HyperLTL are expressively equivalent.

At the time of the submission we became aware of two accepted (not yet published) papers~\cite{lics2021,cav2021} that address the problem of expressing asynchronous variants of information-flow security properties. In both submissions the authors introduce extensions of HyperLTL to address different approaches to deal with asynchronicity.
These papers confirm that the need for a framework enabling a systematic investigation of information-flow properties under different assumptions is of timely importance. 
While these works focus on the asynchronicity of system events
and on the decidability of the corresponding model-checking problems,
we instead consider the (a-)synchronicity (and observability) of specification events under two information-flow semantics and investigate the corresponding expressiveness problems.
For example, in the two-state independence property, 
the state transition does not (necessarily) refer to a system transition but \emph{specifies} a change in the dependency graph between variables.

\section{Conclusion}
\label{sec:concl}

In this paper, we studied the formal specification of sequential 
information-flow hyperproperties, 
especially the paradigmatic hyperproperty of two-state independence. 
We formalized several flavours 
of sequential information-flow using Hypertrace Logic, a first-order 
logic with trace and time quantifiers. 
We introduced a new proof technique for reasoning about the  
expressiveness of linear-time specification formalisms for hyperproperties 
such as HyperLTL. 
In particular, we showed that several natural flavours of sequential 
information flow cannot be expressed 
in HyperLTL due to the fixed order of its quantifiers.

The results in this paper indicate the need to study more asynchronous 
classes of hyperlogics. 
These findings seem to be corroborated by very recent works \cite{lics2021,cav2021} on asynchronous and context HyperLTL.
We plan to study also the expressiveness of these formalisms with respect to sequential information flow.

\bibliography{main}

\clearpage
\appendix

\section{Globally Hypertrace Logic}

\subsection{Theorem \ref{thm:inv:global_FOL}}

\emph{Let \(\Prop\) be a finite set of propositional variables and \(\setsetTraces \subseteq 2^{\vals^\omega_{\Prop}}\) be a hyperproperty.
If there exists a globally hypertrace formula \({\varphi \in \GFOL}\) that generates the same set of sets of traces as the hyperproperty,
\(\setsetTraces = \generatedSet{\varphi}\), then
there exists a set of sets of valuations \(\textbf{M} \subseteq 2^{\vals_{\Prop}}\) that generates the point-wise interpretation of the hyperproperty, 
\(\{ \setTraces[0] \setTraces[1] \ldots \ | \ \setTraces \in \setsetTraces  \} = \textbf{M}^{\omega}\), where \(\setTraces[i] = \{ \trace[i]\ |\ \trace \in \setTraces \}\).}

\begin{proof}
Consider an arbitrary  finite set of propositional variables \(\Prop\) and \(\setsetTraces \subseteq 2^{\vals^\omega_{\Prop}}\).
Assume that there exists \(\varphi \in \GFOL\) s.t.\
\((\star)\ \setsetTraces = \{\setTraces \ |\ \setTraces \models \varphi\}\).
By definition of \(\GFOL\), \(\varphi = \forall i\ \psi(i)\).
Let \(\textbf{M}_{\varphi} = \{\setTraces[n] \ | \  \setTraces \models \varphi \tAnd n \in \nat  \}\). 
Clearly, \(\textbf{M}_{\varphi} \subseteq  2^{\vals_{\Prop}}\).
As \(\Prop\) is finite, all elements of \(\textbf{M}_{\varphi}\) are finite. Additionally, there is only a finite number of valuations over \(\Prop\), so \(\textbf{M}_{\varphi}\) is a finite set, too.

First, we prove \(\{ \setTraces{[0]}\setTraces{[1]} \ldots\ |\  \setTraces \in \setsetTraces  \} \subseteq \textbf{M}_{\varphi}^{\omega}\).
Consider an arbitrary \(\setTraces \in \setsetTraces\). 
By our assumption \((\star)\), \(\setTraces \models \varphi\).
Then, by definition of \(\textbf{M}_{\varphi}\), 
\(\{\setTraces[0], \setTraces[1], \ldots \} \subseteq \textbf{M}_{\varphi}\).
Hence \(\setTraces[0] \setTraces[1] \ldots \in \textbf{M}_{\varphi}^{\omega}\).

Now, we prove \(\textbf{M}_{\varphi}^{\omega} \subseteq \{ \setTraces{[0]}\setTraces{[1]} \ldots\ |\  \setTraces \in \setsetTraces  \}\).
Consider arbitrary \(M = M_0 M_1 \ldots\) s.t.\  \(M \in \textbf{M}_{\varphi}^\omega\).
Then, \(\{M_0, M_1, \ldots\} \subseteq \textbf{M}_{\varphi}\) and, by definition of \(\textbf{M}_{\varphi}\), for all \(j\in \nat\), \((\star\star)\) there exists \(\setTraces\models\varphi\) and \(j' \in \nat\) s.t.\ \(M_j = \setTraces[j']\).
Next we define a set of traces  s.t.\ \(\setTraces_M{[0]}\setTraces_M{[1]} \ldots = M\) and \(\setTraces_{M} \in \setsetTraces\).
Wlog, as \(\textbf{M}_{\varphi}\) is finite, 
\(\textbf{M}_{\varphi} = \{M_0, \ldots, M_n \}\).
For all  \(M_i = \{ m_0, \ldots, m_j\} \in \textbf{M}\), we use \(M_i[l] = m_l\) 
with \(0 \leq l \leq j\).
The set of traces \(\setTraces_M\) is defined below:
\[
 \setTraces_M = \{ \trace_k \, | \, 0 \leq k < \|\textbf{M}_{\varphi}\| \tAnd j \in \nat: \trace_k[j] = M_j[k \bmod \|\textbf{M}_{\varphi}\|] \}.
\]

By our assumption \((\star)\), \(\setTraces_{M} \in \setsetTraces\) iff
\(\setTraces_{M} \models \forall i\ \psi(i)\).
By definition of satisfaction for FOL, 
\(\setTraces_{M} \models \varphi\) iff
for all \(j \in \nat\), \({\setTraces_M[j]\models \psi(i)}\).
Consider an arbitrary \(j \in \nat\). 
Given that there exists \(k = \|\textbf{M}\|\) traces, for all \(M \in \textbf{M}_{\varphi}\) we have \(k \leq |M|\). Then, by \(\trace_k[j] = M_j[k \bmod \|\textbf{M}_{\varphi}\|]\) used in the definition of \(\setTraces_{M}\), it follows that \(\setTraces_{M}[j] = M_j\) (i.e.\ all valuations in \(M_j\) occur at least once at time \(j\) in  \(\setTraces_{M}\) and there is nothing else there).
So, by \((\star\star)\), there exists \(\setTraces\) and \(j' \in \nat\) s.t.\ \({\setTraces[j'] \models \psi(i)}\) and  \(\setTraces[j'] = M_j = \setTraces_M[j]\). 
Hence \(\setTraces_{M} \in \setsetTraces\).
\end{proof}

\section{Expressiveness}

\subsection{Lemma \ref{lemma:exp:open}}

\emph{Let \(\LTLC\) be a class of LTL formulas. For all HyperLTL formulas in the HyperLTL extension of  \(\LTLC\), \(\varphi \in \HyperLTLC\), and all sets of traces \(\setTraces\) and \(\setTraces\) that are \((|\Var(\varphi)|, \LTLC)\)-equivalent, \({\setTraces \Equiv_{(|\Var(\varphi)|, \LTLC)} \setTraces'}\), then for all functions \(\wit: \setTraces \rightarrow \setTraces'\) witnessing the equivalence and all assignments \(\traceAssign_{\setTraces}\) and 
\(\traceAssign_{\setTraces'}\) over \(\setTraces\) and \(\setTraces'\), respectively, only with assignments to the set of free variables in \(\varphi\),
\(\Var(\traceAssign_{\setTraces}) = \Var(\traceAssign_{\setTraces'}) = \fr(\varphi)\):
\((\traceAssign_{\setTraces}, 0) \models \varphi \tIff (\wit(\traceAssign_{\setTraces}), 0) \models \varphi; \tAnd
(\traceAssign_{\setTraces'}, 0) \models \varphi \tIff (\wit^{-1}(\traceAssign_{\setTraces'}), 0) \models \varphi.\)}

\begin{proof}
We prove this statement by structural induction on HyperLTL formulas on a class \(\HyperLTLC\).
The class \(\LTLC\) affects only the quantifier-free part of the formula.

\begin{description}

\item[Quantifier-free \(\varphi \in \HyperLTLC\):]
Then, 
\(\text{free}(\varphi) = \mathcal{V}(\varphi)\). Additionally, \(\varphi\) can be interpreted as an LTL formula over the 
set of propositional variables \(\Prop_{\Var(\varphi)}\) with 
\(\varphi \in \LTLC\).
Consider arbitrary set of traces s.t.\
\(\setTraces \Equiv_{(|\mathcal{V}(\varphi)|, \LTLC)} \setTraces'\) with \(\wit\) being a function that witnesses it.
Now, consider an arbitrary \(\traceAssign_{\setTraces}\) and  \(\traceAssign_{\setTraces'}\) over \(\setTraces'\)
and \(\setTraces\), respectively, s.t.\ \(\mathcal{V}(\traceAssign_{\setTraces}) = \mathcal{V}(\traceAssign_{\setTraces'}) = \text{free}(\varphi)\).
By definition of \(\Equiv_{(|\mathcal{V}(\varphi)|, \LTLC)}\), 
\(\flatT{\traceAssign_{\setTraces}} \Equiv_{\LTLC} \flatT{\wit(\traceAssign_{\setTraces})}\)
and, by definition of \(\Equiv_{\LTLC}\):
\((\star)\ \flatT{\traceAssign_{\setTraces}} \models \varphi \tIff \flatT{\wit(\traceAssign_{\setTraces})} \models \varphi.\)
By Proposition~\ref{prop:zipping}, 
\((\traceAssign_{\setTraces},0) \models \varphi\) iff \({\flatT{\traceAssign_{\setTraces}}[0 \ldots]\models \varphi}\); and 
\((\wit(\traceAssign_{\setTraces}), 0)\models \varphi\) iff \(\flatT{\wit(\traceAssign_{\setTraces})}[0 \ldots] \models \varphi\).
Thus, \({(\traceAssign_{\setTraces}, 0) \models \varphi}\) iff \((\wit(\traceAssign_{\setTraces}), 0) \models \varphi\).

Analogously,
\((\traceAssign_{\setTraces'},0) \models \varphi\) iff \((\flatT{\wit^{-1}(\traceAssign_{\setTraces'})},0) \models \varphi\).

\item[Induction case \(\forall \pi\ \varphi\):]
Assume by induction hypothesis (IH) that the statement holds for arbitrary \(\varphi \in \HyperLTLC\). 
Assume that \emph{(i)}
\(\setTraces \Equiv_{(|\Var(\forall \pi\ \varphi)|, \LTLC)} \setTraces'\).
Note that, wlog we can assume that quantifiers bind a variable already occurring in \(\varphi\), i.e.\
 \(|\Var(\forall \pi\ \varphi)| = |\Var(\varphi)|\). Then, \emph{(i')} \({\setTraces \Equiv_{(|\Var(\varphi)|, \LTLC)} \setTraces'}\), and it has the same witnesses as assumption~\emph{(i)}. 

Let \(\wit: \setTraces \rightarrow \setTraces'\) be a function that witnesses \emph{(i)}.
Now, consider arbitrary \(\traceAssign_{\setTraces}\) and \(\traceAssign_{\setTraces'}\), over \(\setTraces\) and \(\setTraces'\), s.t.\ \(\Var(\traceAssign_{\setTraces}) = \Var(\traceAssign_{\setTraces'}) =\fr(\forall \pi\ \varphi) = \fr(\varphi) \setminus\{\pi\}\).

We prove next that:  \((\traceAssign_{\setTraces}, 0) \models \forall \pi \ \varphi\) iff \({(\wit(\traceAssign_{\setTraces}), 0) \models \forall \pi\ \varphi}\).
We start with the \(\Rightarrow\)-direction of the statement.
Assume that \({(\traceAssign_{\setTraces}, 0) \models \forall \pi \ \varphi}\), then by HyperLTL satisfaction:
\((\star)\ \text{for all } \trace \in \setTraces: (\traceAssign_{\setTraces}[\pi \mapsto \trace], 0) \models \varphi.\)
By Definition~\ref{def:hyperLTL:global:equiv}, \(\Var(\wit(\traceAssign_{\setTraces})) = \Var(\traceAssign_{\setTraces})\). Thus, \(\Var(\wit(\traceAssign_{\setTraces})[\pi \mapsto \trace']) =
\Var(\traceAssign_{\setTraces}) \cup \{\pi\} = \fr(\varphi)\). We can apply the (IH), because \({\setTraces \Equiv_{(|\Var(\varphi)|, \LTLC)} \setTraces'}\), \(\wit\) witnesses it,
and for all \(\trace \in \setTraces\) then \(\traceAssign_{\setTraces}[\pi \mapsto \trace]\) is an assignment over \(\setTraces\). So, it follows:
\(\text{for all } \trace \in \setTraces: (\wit(\traceAssign_{\setTraces}[\pi \mapsto \trace]), 0) \models \varphi.\)

Assume towards a contradiction that \((\wit(\traceAssign_{\setTraces}), 0) \not \models \forall \pi\ \varphi\).
Then, by definition of HyperLTL satisfaction:
\(\text{there exists } \trace' \in \setTraces': (\wit(\traceAssign_{\setTraces})[\pi \mapsto \trace'], 0) \not \models \varphi.\)
We can apply the (IH), because
\(\Var(\wit(\traceAssign_{\setTraces})[\pi \mapsto \trace']) =  \fr(\varphi)\),  \emph{(i')} with \(\wit\) being one of its witnesses,
and for all \(\trace' \in \setTraces'\) the \(\wit(\traceAssign_{\setTraces})[\pi \mapsto \trace']\) is an assignment over \(\setTraces'\). So, it follows:
\[\text{there exists } \trace' \in \setTraces': (\wit^{-1}(\wit(\traceAssign_{\setTraces})[\pi \mapsto \trace']), 0) \not \models \varphi.\]
And by Definition~\ref{def:hyperLTL:global:equiv}:
\(\text{there exists } \trace' \in \setTraces': (\wit^{-1}(\wit(\traceAssign_{\setTraces}))[\pi \mapsto \wit^{-1}(\trace')], 0) \not \models \varphi.\)
As \(\wit\) is a bijective function, \((\wit^{-1}(\wit(\traceAssign_{\setTraces})) = \traceAssign_{\setTraces}\), and so:
\(\text{there exists } \trace' \in \setTraces': (\traceAssign_{\setTraces}[\pi \mapsto \wit^{-1}(\trace')], 0) \not\models \varphi.\)
And this is equivalent to:
\[\text{there exists } \trace' \in \setTraces': \trace = \wit^{-1}(\trace') \tAnd (\traceAssign_{\setTraces}[\pi \mapsto\trace], 0) \not\models \varphi.\]
Given that \(\wit\) is a surjective function, then:
\(\text{there exists } \trace \in \setTraces:  (\traceAssign_{\setTraces}[\pi \mapsto\trace], i) \not\models \varphi.\)
This contradicts \((\star)\). So, the \(\Rightarrow\)-direction holds.

We now prove the \(\Leftarrow\)-direction by contra-position.

Assume that \({(\traceAssign_{\setTraces}, 0) \not\models \forall \pi \ \varphi}\), then by HyperLTL satisfaction:
\(\text{there exists } \trace \in \setTraces: (\traceAssign_{\setTraces}[\pi \mapsto \trace], 0) \not\models \varphi.\)
We can apply now the (IH), because \(\Var(\wit(\traceAssign_{\setTraces})[\pi \mapsto \trace']) = \fr(\varphi)\), \emph{(i')} with \(\wit\)  being one of its witnesses,
and for all \(\trace \in \setTraces\) the \(\traceAssign_{\setTraces}[\pi \mapsto \trace]\) is an assignment over \(\setTraces\). Then, it follows:
\(\text{there exists } \trace \in \setTraces: (\wit(\traceAssign_{\setTraces}[\pi \mapsto \trace]), 0) \not\models \varphi.\)
By Definition~\ref{def:hyperLTL:global:equiv}:
\(\text{there exists } \trace \in \setTraces: \trace' = \wit(\trace) \tAnd (\wit(\traceAssign_{\setTraces})[\pi \mapsto \trace']), 0) \not\models \varphi.\)
By \(\wit\) being surjective, it follows:
\(\text{there exists } \trace' \in \setTraces': (\wit(\traceAssign_{\setTraces})[\pi \mapsto \trace']), 0) \not\models \varphi.\)
And by Definition of HyperLTL satisfaction:
\((\wit(\traceAssign_{\setTraces}), 0) \not\models \forall \pi \ \varphi.\)
Thus, the \(\Leftarrow\)-direction of the statement holds, as well.\\
Hence \((\traceAssign_{\setTraces}, 0) \models \forall \pi \ \varphi\) iff \({(\wit(\traceAssign_{\setTraces}), 0) \models \forall \pi\ \varphi}\).


Now, we prove \((\traceAssign_{\setTraces'}, 0)\!\models\!\forall \pi \ \varphi\) iff \({(\wit^{-1}(\traceAssign_{\setTraces'}), 0)\!\models\! \forall \pi\ \varphi}\).
We start with the \(\Rightarrow\)-direction of the statement.

Like in the previous case, 
we assume that \((\traceAssign_{\setTraces'}, 0) \models \forall \pi \ \varphi\), and then we assume towards a contradiction that \((\wit^{-1}(\traceAssign_{\setTraces'}), 0) \not\models \forall \pi\ \varphi\). The proof is analogous to the previous case up to the point that we infer:
\[\text{there exists } \trace \in \setTraces: \trace' = \wit(\trace) \tAnd (\traceAssign_{\setTraces'}[\pi \mapsto\trace'], i) \not\models \varphi.\]
Then, by \(\wit\) being total we get:
\[\text{there exists } \trace' \in \setTraces': (\traceAssign_{\setTraces'}[\pi \mapsto\trace'], i) \not\models \varphi.\]
And this contradicts the assumption that \((\traceAssign_{\setTraces'}, i) \models \forall \pi \ \varphi\).

The \(\Leftarrow\)- direction is analogous to the previous case, as well, up to the point that we infer:
\[\text{there exists } \trace' \in \setTraces': \trace = \wit(\trace') \tAnd (\wit^{-1}(\traceAssign_{\setTraces})[\pi \mapsto \trace]), 0) \not\models \varphi.\]
Then, by \(\wit\) being total we infer:
\[\text{there exists } \trace \in \setTraces:  (\wit^{-1}(\traceAssign_{\setTraces})[\pi \mapsto \trace]), 0) \not\models \varphi.\]
So, \(\wit^{-1}(\traceAssign_{\setTraces}), 0) \not\models \forall \pi \ \varphi\). Hence the \(\Leftarrow\)- direction holds.

\item[Induction case \(\exists \pi\ \varphi\): ] Assume by induction hypothesis (IH) that the statement holds for arbitrary \(\varphi \in \HyperLTLC\).

We assume that \emph{(i)}
\(\setTraces \Equiv_{(|\Var(\exists \pi\ \varphi)|, \LTLC)} \setTraces'\).
Note that, wlog we can assume that quantifiers bind a variable already occurring in \(\varphi\), i.e.\
\(|\Var(\exists \pi\ \varphi)| = |\Var(\varphi)|\). So, \emph{(i')} \(\setTraces \Equiv_{(|\Var(\varphi)|, \LTLC)} \setTraces'\), and it has the same witnesses as assumption~\emph{(i)}. 
 
Let \(\wit: \setTraces \rightarrow \setTraces'\) be a function that witnesses \emph{(i)}.
Now, consider arbitrary \(\traceAssign_{\setTraces}\) and \(\traceAssign_{\setTraces'}\), over \(\setTraces\) and \(\setTraces'\), s.t.\ \(\Var(\traceAssign_{\setTraces}) = \Var(\traceAssign_{\setTraces'}) =\fr(\exists \pi\ \varphi) = \fr(\varphi) \setminus\{\pi\}\).

\begin{description}
\item[We prove: \((\traceAssign_{\setTraces}, 0) \models \exists \pi \ \varphi\) iff \((\wit(\traceAssign_{\setTraces}), 0) \models \exists \pi\ \varphi\).] \hfill

We start with the \(\Rightarrow\)-direction of the statement.

Assume that \((\traceAssign_{\setTraces}, 0) \models \exists \pi \ \varphi\), then by HyperLTL satisfaction:
\[\text{there exists } \trace \in \setTraces: (\traceAssign_{\setTraces}[\pi \mapsto \trace], 0) \models \varphi.\]

By Definition~\ref{def:hyperLTL:global:equiv}, \(\Var(\wit(\traceAssign_{\setTraces})) = \Var(\traceAssign_{\setTraces})\), and thus \(\Var(\wit(\traceAssign_{\setTraces})[\pi \mapsto \trace']) =
\Var(\traceAssign_{\setTraces}) \cup \{\pi\} = \fr(\varphi)\). Then, we can apply the (IH), because \(\setTraces \Equiv_{(|\Var(\varphi)|, \LTLC)} \setTraces'\) with \(\wit\) being one of its witnesses, and  for all \(\trace \in \setTraces\) the \(\traceAssign_{\setTraces}[\pi \mapsto \trace]\) is an assignment over \(\setTraces\). So, we get:
\[\text{there exists } \trace \in \setTraces: (\wit(\traceAssign_{\setTraces}[\pi \mapsto \trace]), 0) \models \varphi.\]
By Definition~\ref{def:hyperLTL:global:equiv},
\[\text{there exists } \trace \in \setTraces, \trace'\! =\!\wit(\trace) \tAnd (\wit(\traceAssign_{\setTraces})[\pi \mapsto \trace'], 0)\! \models\! \varphi.\]
Then, by \(\wit\) being a total function:
\[\text{there exists } \trace' \in \setTraces': (\wit(\traceAssign_{\setTraces})[\pi \mapsto \trace'], 0) \models \varphi.\]
Hence by HyperLTL satisfaction definition:
\((\wit(\traceAssign_{\setTraces}), 0) \models \exists \pi\ \varphi\).

We now prove the \(\Leftarrow\)-direction of the statement.

Assume that \((\wit(\traceAssign_{\setTraces}), 0) \models \exists \pi \ \varphi\), then by HyperLTL satisfaction:
\[\text{there exists } \trace \in \setTraces: (\wit(\traceAssign_{\setTraces}[\pi \mapsto \trace]), 0) \models \varphi.\]

By Definition~\ref{def:hyperLTL:global:equiv}, \(\Var(\traceAssign_{\setTraces}) = \Var(\wit(\traceAssign_{\setTraces}))\), and thus \(\Var(\traceAssign_{\setTraces}[\pi \mapsto \trace']) =
\Var(\wit(\traceAssign_{\setTraces})) \cup \{\pi\} = \fr(\varphi)\). Then, we can apply the (IH), because \(\setTraces \Equiv_{(|\Var(\varphi)|, \LTLC)} \setTraces'\) with \(\wit\) being one of its witnesses, and  for all \(\trace \in \setTraces\) the \(\wit(\traceAssign_{\setTraces}[\pi \mapsto \trace])\) is an assignment over \(\setTraces'\). So, we get:
\[\text{there exists } \trace \in \setTraces: (\traceAssign_{\setTraces}[\pi \mapsto \trace], 0) \models \varphi.\]
And by HyperLTL satisfaction definition, \((\traceAssign_{\setTraces}, 0) \models \exists \pi \ \varphi\).

\item[We prove: \((\traceAssign_{\setTraces'}, 0) \models \exists \pi \ \varphi\) iff \((\wit^{-1}(\traceAssign_{\setTraces'}), 0) \models \exists \pi\ \varphi\).] \hfill

We start with the \(\Rightarrow\)-direction of the statement. It is analogous to the previous case, up to the point that we infer:
\(\text{there exists } \trace' \in \setTraces': \trace = \wit^{-1}(\trace') \tAnd (\wit^{-1}(\traceAssign_{\setTraces'})[\pi \mapsto\trace], 0) \models \varphi.\)
Then, by \(\wit\) being surjective we get:
\[\text{there exists } \trace \in \setTraces: (\wit^{-1}(\traceAssign_{\setTraces'})[\pi \mapsto\trace], 0) \models \varphi.\]
Hence by HyperLTL satisfaction definition:
\((\wit^{-1}(\traceAssign_{\setTraces'}), 0) \models \exists \pi\ \varphi\).

The \(\Leftarrow\)-direction is analogous to the previous case. \qedhere
\end{description}
\end{description}
\end{proof}

\subsection{Proposition \ref{prop:ltlbox:equiv}}
\emph{For all all two traces \(\trace\)
and~\(\trace'\),
\(\trace \Equiv_{\LTLBox} \trace'\) iff, for all \(\varphi \in \LTLBox\),
\(\trace \models \varphi\) iff \(\trace' \models \varphi\).}

\begin{proof}
Consider arbitrary traces \(\trace\) and \(\trace'\).

\(\Rightarrow:\)
Assume that \(\trace \Equiv_{\LTLBox} \trace'\). 
Consider an arbitrary \(\varphi \in \LTLBox\), so
\(\varphi = \Globally \psi\) where \(\psi\) is a propositional formula.
Then, 
\(\trace \not\models \varphi\) iff there exists  \(i \in \nat\) s.t.\ 
\(\trace[i] \not\models \psi\).
And this is equivalent to, there exists \(\val \in \{\trace[i] \ | \ i \in \nat\}\) s.t.\
\(\val \not\models \psi\).
By \(\trace \Equiv_{\LTLBox} \trace'\), \(\val \in \{\trace[i] \ | \ i \in \nat\}\) iff \(\val \in \{\trace'[j] \ | \ j \in \nat\}\).
Then from an analogous reasoning, the former is equivalent to \(\trace' \not\models \varphi\).

\(\Leftarrow:\) Assume that for all \(\varphi \in \LTLBox\) we have 
\((\star)\) \(\trace \models \varphi\) iff \(\trace' \models \varphi\).
Now, assume towards a contradiction that \(\trace \not\Equiv_{\LTLBox} \trace'\).
Consider first the case that there exists \(i\in \nat\) s.t.\ \(\trace[i] \notin \{\trace'[j] \ | \ j \in \nat\}\). This contradicts \((\star)\), because it entails that there exists \(\varphi \in \LTLBox\) with \(\varphi =\Globally \psi\) s.t.\ 
\(\trace[i] \models \psi\) and for all \(j\in \nat\) \(\trace[j] \not\models \psi\).
The case that there exists \(j\in \nat\) s.t.\ \(\trace'[j] \notin \{\trace'[i] \ | \ i \in \nat\}\) is analogous.
\end{proof}

\subsection{Theorem \ref{thm:time:indist}}
\emph{For all time-prefixed Hypertrace sentences \(\varphi \in \TimeTFOL\) and all sets of traces, \(\setTraces\) and \(\setTraces'\), that are \(|\timeVars(\varphi)|\)-point equivalent, \(\setTraces \Equiv^{\syncS}_{|\timeVars(\varphi)|} \setTraces'\), where  \(|\timeVars(\varphi)|\) is the number of time variables in \(\varphi\), then 
\(\setTraces \models \varphi\) iff \(\setTraces' \models \varphi\).}
\begin{proof}
We evaluate time-prefixed formulas as in FOL with sorts. We denote by \((\traceAssign_{\nat}, \traceAssign_{\setTraces})\) a pair with the assignments for variables over the sort time and the sort of traces, respectively.
Wlog, we can assume that the variables can be identified by the position they are quantified. Then, given a \(k\)-point equivalence and an assignment over time, \(\traceAssign_{\nat}\) that has assignments for the variables \(\timeVars(\traceAssign_{\nat}) = \{\traceVar_1, \ldots \traceVar_n\}\) with \(n \leq k\), then there exists a witness function for the tuple \((\traceVar_1, \ldots, \traceVar_n)\) which we denote by \(\wit_{\traceAssign_{\nat}}\).
We prove the theorem by proving the following lemma first:

\emph{For all time-prefixed Hypertrace formulas \(\varphi \in \TimeTFOL\) and all sets of traces, \(\setTraces\) and \(\setTraces'\), that are \(|\timeVars(\varphi)|\)-point equivalent, \(\setTraces \Equiv^{\syncS}_{|\timeVars(\varphi)|} \setTraces'\), where  \(|\timeVars(\varphi)|\) is the number of time variables in \(\varphi\), then for all time assignments \(\traceAssign_{\nat}\), and for all trace assignments \(\traceAssign_{\setTraces}\) and \(\traceAssign_{\setTraces'}\), which are
over \(\fr(\varphi)\), we have:
\((\traceAssign_{\nat}, \traceAssign_{\setTraces}) \models \varphi\) iff \((\traceAssign_{\nat}, \wit_{\traceAssign_{\nat}}(\traceAssign_{\setTraces})) \models \varphi\); and \((\traceAssign_{\nat}, \traceAssign_{\setTraces'}) \models \varphi\) iff \((\traceAssign_{\nat}, \wit_{\traceAssign_{\nat}}^{-1}(\traceAssign_{\setTraces})) \models \varphi\);}

We prove this lemma by structural induction on time-prefixed formulas.
The induction step for the time prefix part is trivial, as the assignment over time variables in both sides of the implication is the same. For the trace quantifier part, the proof in analogous to Lemma \ref{lemma:exp:open}.
The only difference is the base case, that follows from the definition of \(k\)-point equivalence.
\end{proof}

\subsection{Point Semantics}

\subsubsection{Lemma \ref{lemma:step:settraces_and_prop}}
\(\setTraces^{\pointW}_{n} \in  \setsetTraces^1_{\syncS}\) and 
\({\setTraces'}^{\pointW}_{n} \notin  \setsetTraces^1_{\syncS}\). 

\begin{proof}
\(\setTraces \in \setsetTraces^1_{\syncS}\) iff
for all \(i \in \nat\):
\[\setTraces[i]\! \in\! \{ \M\, |\, \M \models  (x(\traceVar,i) \leftrightarrow x(\traceVar_{\exists},i)) \wedge 
(y(\traceVar',i) \leftrightarrow y(\traceVar_{\exists},i)).\}\]
Consider arbitrary \(n\in \nat\). We represent the valuations in \((x,y)\).

For all \(n+1 < i\), 
\(\setTraces^{\pointW}_{n}[i] = \setTraces'^{\pointW}_{n}[i] = \{00\}\). 

For all \( i < n+1\), 
\(\setTraces^{\pointW}_{n}[i] = \setTraces'^{\pointW}_{n}[i] = \{ 00,01,10,11\}\). 

As \(\setTraces^{\pointW}_{n}[n+1] = \{ 00,01,10,11\}\),
then \(\setTraces^{\pointW}_{n}\in \setsetTraces^1_{\syncS}\).
And, as \(\setTraces'^{\pointW}_{n}[n+1] =\{ 10,00\}\), then 
\(\setTraces^{\pointW}_{n}\notin \setsetTraces^1_{\syncS}\).
\end{proof}

\subsubsection{Theorem \ref{thm:time:point:sync}}

\emph{Consider the following time-prefixed hypertrace formula:
\[
\begin{split}
\varphi^{\sync}_{\text{time}} \overset{\text{def}}{=} & \ \exists j \forall i < j \forall k \leq j \forall \traceVar \forall \traceVar' \exists \traceVar_{\exists} \\
&\ \ \big(
\neg a(\traceVar, i) \wedge \neg a(\traceVar', i)  \wedge
(x(\traceVar,i) \leftrightarrow x(\traceVar_{\exists},i)) \wedge 
(y(\traceVar',i) \leftrightarrow y(\traceVar_{\exists},i))\big) \wedge \\
&\ \ \big(
a(\traceVar, j) \wedge a(\traceVar', j)  \wedge 
(x(\traceVar,k) \leftrightarrow x(\traceVar'_{\exists},k)) \wedge 
(z(\traceVar',k) \leftrightarrow z(\traceVar'_{\exists},k))\big)
\end{split}
\]
Then, \(\generatedSet{\varphi^{\sync}_{\text{time}}} = \setsetTraces^{\sync}_{\syncS}\).}

\begin{proof}
Note that
\(\setTraces \models \varphi^{\sync}_{\text{time}}\) iff
\(\setTraces \models  \exists j \forall i < j  \forall \traceVar \,
\neg a(\traceVar, i) \wedge a(\traceVar, j).\) Thus, it only includes sets of traces with synchronous action.

Additionally, for all \(\setTraces\) that have synchronous action:
\begin{align*}
\setTraces \models \varphi^{\sync}_{\text{time}} \tIff &
\setTraces \models \exists j \forall i < j \forall \traceVar \forall \traceVar' \exists \traceVar_{\exists} \\
& \qquad \big(
\neg a(\traceVar, i) \wedge \neg a(\traceVar', i)  \wedge
(x(\traceVar,i) \leftrightarrow x(\traceVar_{\exists},i)) \wedge 
(y(\traceVar',i) \leftrightarrow y(\traceVar_{\exists},i))\big) \wedge 
a(\traceVar, j) \\
& \tAnd\\
&\setTraces \models \exists j \forall k \leq j \forall \traceVar \forall \traceVar' \exists \traceVar_{\exists} \\
&\qquad  
a(\traceVar, j) \wedge a(\traceVar', j)  \wedge 
(x(\traceVar,k) \leftrightarrow x(\traceVar'_{\exists},k)) \wedge 
(z(\traceVar',k) \leftrightarrow z(\traceVar'_{\exists},k))\big)\\
  \tIff & \setTraces[\DDD a]\models \forall i 
 \forall \traceVar  \forall \traceVar'  \exists \traceVar_{\exists}
\big(\timeDef(\traceVar, i)  \wedge \timeDef(\traceVar', i) \big) \rightarrow\\
&\qquad \big(\timeDef(\traceVar_{\exists}, i) \wedge (x(\traceVar,i) \leftrightarrow x(\traceVar_{\exists},i)) \wedge 
(y(\traceVar',i) \leftrightarrow y(\traceVar_{\exists},i))\\
& \tAnd\\
& \setTraces[a\DDD]\models \forall k 
 \forall \traceVar  \forall \traceVar'  \exists \traceVar_{\exists}
\big(\timeDef(\traceVar, k)  \wedge \timeDef(\traceVar', k) \big) \rightarrow\\
&\qquad \big(\timeDef(\traceVar_{\exists}, k) \wedge (x(\traceVar,k) \leftrightarrow x(\traceVar_{\exists},k)) \wedge 
(z(\traceVar',k) \leftrightarrow z(\traceVar_{\exists},k)).
\end{align*}
\end{proof}

\subsubsection{Lemma \ref{lemma:step:settraces:equiv}} 
\emph{For all \(n\in \nat\),
\(\setTraces^{\pointW}_{n} \Equiv_{(n, \LTLBox)} {\setTraces'}^{\pointW}_{n}\).}

\begin{proof}
Consider arbitrary \(n \in \nat\).
We define the witness function \(\wit_n: \setTraces^{\pointW}_{n} \rightarrow \setTraces'^{\pointW}_{n}\) below:
\begin{equation*}
\wit_n(\trace) = 
\begin{cases}
(00)^{n}\ 10\ 00^\omega & \tIf \trace = (00)^{n}\ 10\ 10\ 00^\omega\\
(00)^{n}\ 01 \ 00^\omega
& \tIf \trace = (00)^{n}\ 01\ 01\ 00^\omega\\
\trace & \text{ otherwise.}
\end{cases}
\end{equation*}
Clearly, this function is both bijective and total.

By definition of \(\wit_n\) and \({\setTraces}^{\pointW}_n\), then for all assignments of size \(n\) over it, \(\traceAssign_{{\setTraces}^{\pointW}_n}\):
\begin{enumerate}[(a)]
\item for all \(i \neq n+1\), 
\(\flatT{\traceAssign_{{\setTraces}^{\pointW}_n}}[i] = \flatT{\wit_n(\traceAssign_{{\setTraces}^{\pointW}_n})}[i]\); and

\item for all \(\traceVar \in \Var\), if  
\(\traceAssign_{{\setTraces}^{\pointW}_n}(\traceVar) \notin \{(00)^{n}\ 10\ 10\ 00^\omega,\) \((00)^{n}\ 01\ 01\ 00^\omega\}\), then 
\(\traceAssign_{{\setTraces}^{\pointW}_n}(\traceVar) = \wit_n(\traceAssign_{{\setTraces}^{\pointW}_n}(\traceVar))\).
\end{enumerate}
Analogously for all assignments over \({\setTraces'}^{\pointW}_n\) of size \(n\) and \(\wit^{-1}\).

It follows from the definition of \({\setTraces}^{\pointW}_n\), that for all assignments of size \(m < n\) over \({\setTraces}^{\pointW}_n\), \(\traceAssign_{{\setTraces}^{\pointW}_n}^{m}\), there exists \(0 \leq k < n\) s.t.\ 
for all \(\traceVar \in \Var\), \(\traceAssign_{{\setTraces}^{\pointW}_n}^{m}(\traceVar) \notin \{(00)^k\ 10 \ 00^\omega, (00)^k\ 01\ 00^\omega\}\). Then, \((\dagger)\)  \(\traceAssign_{{\setTraces}^{\pointW}_n}^{m}[k] \in \{\{11\}, \{00\}, \{11, 00\}\}\), because the only way to get valuations \(10\) and \(10\) at time \(k\) is with the missing traces. 

Consider arbitrary \(n\in \nat\) and  assignment over \({\setTraces}^{\pointW}_n\) of size \(n\), \(\traceAssign_{{\setTraces}^{\pointW}_n}\).
If for all \(\traceVar \in \Var\), \(\traceAssign_{{\setTraces}^{\pointW}_n}(\traceVar) \notin \{(00)^{n}\ 10\ 10\ 00^\omega,\) \((00)^{n}\ 01\ 01\ 00^\omega\}\), then by (b), for all \(i\in \nat\), \(\flatT{\wit_n(\traceAssign_{{\setTraces}^{\pointW}_n})}[i] = \flatT{\traceAssign_{{\setTraces}^{\pointW}_n}}[i]\).
Now we assume that there exists 
 \(Y = \{\traceVar_0, \ldots, \traceVar_l\}\), with \(0 \leq l < n\) s.t.\ \(\traceAssign_{{\setTraces}^{\pointW}_n}(\traceVar) \in \{(00)^{n}\ 10\ 10\ 00^\omega,\) \((00)^{n}\ 01\ 01\ 00^\omega\}\), with \(\traceVar \in Y\), and
for all \(\traceVar \notin Y\),  \(\traceAssign_{{\setTraces}^{\pointW}_n}(\traceVar) \notin \{(00)^{n}\ 10\ 10\ 00^\omega,\) \((00)^{n}\ 01\ 01\ 00^\omega\}\).
We can prove from \((\dagger)\) that  there exists \(k\) s.t.\
\(\flatT{\traceAssign_{{\setTraces}^{\pointW}_n}|_Y}[n+1] = \flatT{\traceAssign_{{\setTraces}^{\pointW}_n}|_Y}[k]\), where \(\traceAssign_{{\setTraces}^{\pointW}_n}|_Y\) is  \(\traceAssign_{{\setTraces}^{\pointW}_n}\) without the assignments to the variables in \(Y\). Moreover, it follows as well, that for all \(\traceVar \in Y\), 
\(\flatT{\traceAssign_{{\setTraces}^{\pointW}_n}}[k] = 00\), and so
there exists \(k\) s.t.\ 
\(\flatT{\wit_n(\traceAssign_{{\setTraces}^{\pointW}_n})}[n+1] = \flatT{\traceAssign_{{\setTraces}^{\pointW}_n}}[k]\).
Thus, \((\star)\) for all assignments over \({\setTraces}^{\pointW}_n\) of size \(n\), \(\traceAssign_{{\setTraces}^{\pointW}_n}\) and all \(i\in \nat\) there exists \(j\in \nat\) s.t.\ \(\flatT{\wit_n(\traceAssign_{{\setTraces}^{\pointW}_n})}[i] = \flatT{\traceAssign_{{\setTraces}^{\pointW}_n}}[j]\).

By \(\flatT{\traceAssign_{{\setTraces}^{\pointW}_n}}[n+1] = \flatT{\wit_n(\traceAssign_{{\setTraces}^{\pointW}_n})}[n]\) and (a), then \((\star\star)\) for all assignments of size \(n\) over \({\setTraces}^{\pointW}_n\), \(\traceAssign_{{\setTraces}^{\pointW}_n}\) and for all \(i\in \nat\) there exists \(j\in \nat\) s.t.\ 
\(\flatT{\traceAssign_{{\setTraces}^{\pointW}_n}}[i] = \flatT{\wit_n(\traceAssign_{{\setTraces}^{\pointW}_n})}[j]\).

By \((\star)\)  and \((\star\star)\), for all assignments of size \(n\) over \({\setTraces}^{\pointW}_n\), \(\traceAssign_{{\setTraces}^{\pointW}_n}\), 
we have 
\(\flatT{\traceAssign_{{\setTraces}^{\pointW}_n}}[i] \Equiv_{\LTLBox} \flatT{\wit_n(\traceAssign_{{\setTraces}^{\pointW}_n})}[j]\).

We prove analogously that 
for all assignments of size \(n\) over \({\setTraces'}^{\pointW}_n\), \(\traceAssign_{{\setTraces'}^{\pointW}_n}\), 
we have 
\(\flatT{\traceAssign_{{\setTraces'}^{\pointW}_n}}[i] \Equiv_{\LTLBox} \flatT{\wit^{-1}_n(\traceAssign_{{\setTraces'}^{\pointW}_n})}[j]\).
\end{proof}

\subsubsection{Lemma \ref{lemma:step:settraces:equiv}}
\emph{For all \(n\in \nat\),
\(\setTraces^{\pointW}_{n} \Equiv_{(n, \LTLBox)} {\setTraces'}^{\pointW}_{n}\).}

\begin{proof}
Consider arbitrary \(n \in \nat\).
We define the witness function \(\wit_n: \setTraces^{\pointW}_{n} \rightarrow \setTraces'^{\pointW}_{n}\) below:
\begin{equation*}
\wit_n(\trace) = 
\begin{cases}
(00)^{n}\ 10\  (00)^\omega & \tIf \trace = (00)^{n}\ 10\ 10\  (00)^\omega\\
(00)^{n}\ 01 \  (00)^\omega
& \tIf \trace = (00)^{n}\ 01\ 01\  (00)^\omega\\
\trace & \text{ otherwise.}
\end{cases}
\end{equation*}
Clearly, this function is both bijective and total.

We need to prove that for all assignments over \(\setTraces^{\pointW}_{n}\) and \({\setTraces'}^{\pointW}_{n}\)  of size \(k\), \(\traceAssign_{\setTraces^{\pointW}_{n}}\) and \(\traceAssign_{{\setTraces'}^{\pointW}_{n}}\), we have:
\({\flatT{\traceAssign_{\setTraces^{\pointW}_{n}}} \Equiv_{\LTLBox} \flatT{\wit(\traceAssign_{\setTraces^{\pointW}_{n}})}}\) and
\(\flatT{\traceAssign_{{\setTraces'}^{\pointW}_{n}}} \Equiv_{\LTLBox} \flatT{\wit^{-1}(\traceAssign_{{\setTraces'}^{\pointW}_{n}})}\).
Recall that \({\flatT{\traceAssign_{\setTraces^{\pointW}_{n}}} \Equiv_{\LTLBox} \flatT{\wit(\traceAssign_{\setTraces^{\pointW}_{n}})}}\) iff
\(\{\flatT{\traceAssign_{\setTraces^{\pointW}_{n}}}[i] \ | \ i \in \nat\} = \{\flatT{\wit(\traceAssign_{\setTraces^{\pointW}_{n}})}[j] \ | \ j \in \nat\}\).



The interesting case is the time \(n+1\) for assignments to the traces that are different in \({\setTraces}^{\pointW}_n\) and \({\setTraces'}^{\pointW}_n\). 
We show below that for all assignment over \({\setTraces}^{\pointW}_n\) of size \(n\), there exists a time \(k\) that has the same valuations in the flattened assignment at time \(n+1\).

Consider an
arbitrary \(n\in \nat\) and  assignment over \({\setTraces}^{\pointW}_n\) of size \(n\), \(\traceAssign_{{\setTraces}^{\pointW}_n}\).
Assume there exists
 \(Y = \{\traceVar_0, \ldots, \traceVar_l\}\), with \(0 \leq l < n\) s.t.\ \(\traceAssign_{{\setTraces}^{\pointW}_n}(\traceVar) \in \{(00)^{n}\ 10\ 10\  (00)^\omega,\) \((00)^{n}\ 01\ 01\  (00)^\omega\}\), with \(\traceVar \in Y\), and
for all \(\traceVar \notin Y\),  \(\traceAssign_{{\setTraces}^{\pointW}_n}(\traceVar) \notin \{(00)^{n}\ 10\ 10\  (00)^\omega,\) \((00)^{n}\ 01\ 01\  (00)^\omega\}\).

For all assignments of size \(m < n\) over \({\setTraces}^{\pointW}_n\), \(\traceAssign_{{\setTraces}^{\pointW}_n}^{m}\), there exists \(0 \leq k < n\) s.t.\ 
for all \(\traceVar \in \Var\), \(\traceAssign_{{\setTraces}^{\pointW}_n}^{m}(\traceVar) \notin \{(00)^k\ 10 \ (00)^\omega, (00)^k\ 01\ (00)^\omega\}\). Note that we have \(n\) possible combinations for such pairings. Then, \((\bullet)\)  \(\traceAssign_{{\setTraces}^{\pointW}_n}^{m}[k] \in \{\{11\}, \{00\}, \{11, 00\}\}\), because the only way to get valuations \(10\) and \(10\) at time \(k\) is with the missing traces. 
Then,  there exists \(k\) s.t.\
\(\flatT{\traceAssign_{{\setTraces}^{\pointW}_n}|_Y}[n+1] = \flatT{\traceAssign_{{\setTraces}^{\pointW}_n}|_Y}[k]\), where \(\traceAssign_{{\setTraces}^{\pointW}_n}|_Y\) is  \(\traceAssign_{{\setTraces}^{\pointW}_n}\) without the assignments to the variables in \(Y\). Moreover, it follows as well, that for all \(\traceVar \in Y\), 
\(\flatT{\traceAssign_{{\setTraces}^{\pointW}_n}}[k] = 00\). So, 
there exists \(k\) s.t.\ 
\(\flatT{\wit_n(\traceAssign_{{\setTraces}^{\pointW}_n})}[n+1] = \flatT{\traceAssign_{{\setTraces}^{\pointW}_n}}[k]\).
Thus, \((\star)\) for all assignments over \({\setTraces}^{\pointW}_n\) of size \(n\), \(\traceAssign_{{\setTraces}^{\pointW}_n}\) and all \(i\in \nat\) there exists \(j\in \nat\) s.t.\ \(\flatT{\wit_n(\traceAssign_{{\setTraces}^{\pointW}_n})}[i] = \flatT{\traceAssign_{{\setTraces}^{\pointW}_n}}[j]\).
\end{proof}

\subsection{Segment Semantics}

\subsubsection{Theorem \ref{thm:sync:sync:overall}}

\emph{Consider the following HyperLTL formula: 
\[\varphi^{\sync}_{\syncO} \overset{\text{def}}{=}
\forall \traceVar \forall \traceVar' \exists \traceVar_{\exists} \exists \traceVar'_{\exists}\, 
(\neg a_{\traceVar} \wedge \neg a_{\traceVar'} \wedge x_{\traceVar} = x_{\traceVar{\exists}} \wedge y_{\traceVar'} = y_{\traceVar_{\exists}}\!) \Until
(a_{\traceVar} \wedge a_{\traceVar'} \wedge \Box (x_{\traceVar} = x_{\traceVar_{\exists}} \wedge z_{\traceVar'} = z_{\traceVar'_{\exists}}))\]
Then, \(\generatedSet{\varphi^{\sync}_{\syncO}} =   \setsetTraces^{\sync}_{\syncO}\).}

\begin{proof}
Note that \(x(\traceVar,i) \leftrightarrow x(\traceVar_{\exists},i)\)
in Hypertrace Logic corresponds to \(x_{\traceVar} = x_{\traceVar_{\exists}}\) in HyperLTL.
By Definition~\ref{def:prop:two_state_indep}, \(\setTraces \in \setsetTraces^{\sync}_{\syncO}\) iff:
(i) \({\setTraces \models \exists i\ \forall \traceVar\ \text{min}(\trace,a,i)}\);
(ii) \(\setTraces[\ldots a] \models \ind{x}{y}{\syncO}\); and  
(iii) \(\setTraces[a \ldots] \models \ind{x}{z}{\syncO}\).
Then, by HyperLTL satisfaction, for all set of traces~\(\setTraces\):
\begin{align*}
&\setTraces\! \models \varphi^{\sync}_{\syncO} \tIff \\
&\setTraces\! \models\!\! \forall \traceVar \forall \traceVar'  
(\neg a_{\traceVar} \wedge \neg a_{\traceVar'} ) \Until(a_{\traceVar} \wedge a_{\traceVar'}),\\
&\setTraces\! \models \!\! \forall \traceVar \forall \traceVar' \exists \traceVar_{\exists}\, 
(\neg a_{\traceVar} \wedge \neg a_{\traceVar'} \wedge x_{\traceVar} = x_{\traceVar_{\exists}} \wedge y_{\traceVar'} = y_{\traceVar_{\exists}}) \Until
(a_{\traceVar} \wedge a_{\traceVar'}), \tAnd\\
&\setTraces\! \models\!\! \forall \traceVar \forall \traceVar' \exists \traceVar'_{\exists}\, 
(\neg a_{\traceVar} \wedge \neg a_{\traceVar'}) \Until\
(a_{\traceVar} \wedge a_{\traceVar'} \wedge \Box (x_{\traceVar} = x_{\traceVar_{\exists}} \wedge z_{\traceVar'} = z_{\traceVar'_{\exists}})).
\end{align*}
\noindent We can prove, by satisfaction for HyperLTL and Hypertrace Logic formulas, that:
\begin{align*}
&\setTraces\! \models_{H} \forall \traceVar \forall \traceVar'  
(\neg a_{\traceVar} \wedge \neg a_{\traceVar'} ) \Until(a_{\traceVar} \wedge a_{\traceVar'}) \tIff \setTraces\! \models \exists i\, \forall \traceVar\
a(\traceVar,i)  \wedge
\forall 0 \leq j < i\ 
\neg a(\traceVar,j)
\end{align*}
Hence  \(\generatedSet{\varphi^{\sync}_{\syncO}} =   \setsetTraces^{\sync}_{\syncO}\).
\end{proof}

\subsubsection{Lemma \ref{lemma:async:sync}}

\(\setTraces_n^{\async} \in \setsetTraces^{\async}_{\syncO}\),  \({\setTraces'}_n^{\async} \not\in \setsetTraces^{\async}_{\syncS}\) and
\({\setTraces'}_n^{\async}|_{a} \not\in \setsetTraces^{\novisible}_{\syncS}\).
\begin{proof}
We start by proving \(\setTraces_n^{\async} \in \setsetTraces^{\async}_{\syncO}\).

First, we prove that \(\setTraces_n^{\async}[\ldots a] \models \ind{x}{y}{\syncO}\).
By definition of slicing of sets of traces:
\begin{align*}
\setTraces_n^{\async}[\ldots a] = \{ 0000, 0010, (0000)^{n+4}, (0010)^{n+4}\}.
\end{align*}
Then, by Definition~\ref{def:independence}, 
 \(\setTraces_n^{\async}[\ldots a] \models \ind{x}{y}{\syncO}\) holds because we can choose 
 \(\pi_\exists = \pi'\).

Now, we prove that \(\setTraces_n^{\async}[a \ldots] \models \ind{x}{z}{\syncO}\).
By definition of slicing of sets of traces, 
\(\setTraces_n^{\async}[a \ldots] = \{ \trace_0 (1000)^{\omega}, \trace_1 (1000)^{\omega}, \trace_0 (1110)^{\omega}, \trace_1\ (1000)^{n+4}\ (1110)^{\omega}\}\), where \(\trace_0\) and \(\trace_1\) are as in Definition~\ref{def:models:async_action}.
Then, as in the previous case, we can choose  \(\pi_\exists = \pi'\) to show that
 \(\setTraces_n^{\async}[\ldots a] \models \ind{x}{z}{\syncO}\) holds.

We prove now that \({\setTraces'}_n^{\async} \not\in \setsetTraces^{\async}_{\syncS}\).

We show that
 \({\setTraces'}_n^{\async}[\ldots a] \not\models \ind{x}{z}{\syncS}\).
By Definition~\ref{def:models:async_action} and definition of slicing:
\begin{align*}
{t'}_1[a\ldots] &= t_1[1] t_1[2] \ldots t_1[2n+10] t_1[2n+12] \ldots \\
                &= \tau_1[0] \tau_1[1] \ldots \tau_1[2n+9] \tau_1[2n+11] \ldots \\
 				&= {t'}_2[a\ldots]\\
{t'}_3[a\ldots] &= t_3[n+4] t_3[n+5] \ldots t_3[2n+10] t_3[2n+12] \ldots\\
				&= \tau_0[0] \tau_1[1] \ldots \tau_0[n+6] \tau_0[n+8] \ldots \\
 				&= {t'}_4[a\ldots]
\end{align*}
Note that, \(2n+10-(n+4) = n+6\).

Note that \((\star)\) \(({t'}_1[a\ldots])[2n+9] = ({t'}_2[a\ldots])[2n+9] = 1000\)
and \(({t'}_3[a\ldots])[2n+9] = ({t'}_4[a\ldots])[2n+9] = 1111\).
If we chose \(i=2n+9\), \(\pi = {t'}_3\) and \(\pi' = {t'}_1\), then
there should exist a trace \(t_{\exists} \in {\setTraces'}_n^{\async}\) s.t.\ 
\(({t}_{\exists}[a\ldots])[2n+9](x) = ({t'}_3[a\ldots])[2n+9](x) = 1\) and
\(({t}_{\exists}[a\ldots])[2n+9](z) = ({t'}_1[a\ldots])[2n+9](z) = 0\).
However, by \((\star)\) we know that there is not such trace in \({\setTraces'}_n^{\async}\).
Hence \({\setTraces'}_n^{\async} \not\in \setsetTraces^{\async}_{\syncS}\).

The set of set of traces \({\setTraces'}_n^{\async}|_{a}\) is the set \({\setTraces'}_n^{\async}\) where all valuations of \(a\) are removed. 
We need to prove that there is no extension of \({\setTraces'}_n^{\async}|_{a}\) with (possibly new) valuations in \(a\) that makes it an element of \(\setsetTraces^{\novisible}_{\syncS}\).
We will abstract the extension of \({\setTraces'}_n^{\async}|_{a}\) by defining a function
\({g: {\setTraces'}_n^{\async}|_{a} \rightarrow \nat}\) that given a set of traces in \({\setTraces'}_n^{\async}|_{a}\) returns the index where \(a\) first holds. We then redefine the slicing operator to slice w.r.t.\ this function, as follows:
\(\setTraces[\ldots g] = \{ \trace[\ldots g(\trace)]\ |\ \trace \in \setTraces \}\).

We refer to the elements of \({\setTraces'}_n^{\async}|_{a}\) by the same names as in the definition of \({\setTraces'}_n^{\async}\).
By construction of \({\setTraces'}_n^{\async}|_{a}\), the function \(g\) needs to guarantee the following conditions for \(\setTraces[\ldots g] \models \ind{x}{y}{\pointW}\) to hold:
\(g(t'_1) \leq 5n+23\) and  \(g(t'_3) \leq 5n+23\), because \(t'_1[5n+23] = 000 = t'_3[5n+23]\) and \(t'_2[4n+23] = 111 = t'_4[5n+23]\). So, we are missing valuations \(10\) and \(01\) in \((x,y)\), to prove the independence of \(y\) w.r.t.\ \(x\).

If \(g(t'_1) = g(t'_2) = g(t'_3) =g(t'_4) =1\), then 
\(\setTraces[g\ldots] \not\models \ind{x}{z}{\pointW}\), because 
\({\setTraces'}_n^{\async}|_{a}[1] = \{000,010,111\}\) and so we are missing the valuation \(01\) in \((x,z)\).

We proceed by case analysis.

\textbf{Case \(g(t'_3) = g(t'_4) = 1\):}
We show below the first \(n+4\) steps of the slice of \(t'_3\) and \(t'_4\):
\begin{align*}
&\trace_3'[1 \ldots n+4] =  
(000)^{n+3}\ 110\\
&\trace_4'[1 \ldots n+4] =  
(010)^{n+3}\ 110
\end{align*}
To find a compatible slicing of \(t'_1\) and \(t'_2\) we need it to satisfy the following:
\begin{itemize}
    \item for the first \(n+3\) we can only have the valuation \(00\) in \((x,z)\), as there is no time point where we can get at the same time \(10\) and \(11\);
    \item at the \(n+4\) we cannot have \(01\) as it is not possible with only one trace left cover all the valuations missing (\(00\) and \(11\)).
\end{itemize}

Then, the time \(n+7\) is the only slicing of \(t'_1\) and \(t'_2\) that satisfies this conditions and guarantees that
\(x\) is independent of \(z\) for the first \(n+4\) elements of the slicing suffix, as we show below:
\begin{align*}
&\trace_1'[n+7 \ldots 2n+10] =  
(000)^{n+3}\ 110\\
&\trace_2'[n+7 \ldots 2n+10] =  
(000)^{n+3}\ 110\\
&\trace_3'[1 \ldots n+4] =  
(000)^{n+3}\ 110\\
&\trace_4'[1 \ldots n+4] =  
(010)^{n+3}\ 110
\end{align*}

However, if \(g(t'_1) = g(t'_2) = n+7\), then
\(\trace'_1[4n+20] = \trace'_2[4n+20] = 000\) while
\(\trace'_3[3n+13] = \trace'_4[3n+13] = 111\).
So, we are missing valuations \(01\) and \(10\) in \((x,z)\).
Hence, for \(g(t'_3) = g(t'_4) = n+7\), 
\(\setTraces[g\ldots] \not\models \ind{x}{z}{\pointW}\).
So, \(g(t'_3) = g(t'_4) > 1\).

\textbf{Case \(g(t'_1) = g(t'_2) > 1\):}
As \(g(t'_3) = g(t'_4) > 1\), then the prefix of a slicing with \(g(t'_1) = g(t'_2) > 1\) does not satisfy \(\ind{x}{y}{\pointW}\).
Note that \(t'_1[1] = t'_2[1] = 111\) while \(t'_3[1] = 000\) and  \(t'_4[1] = 010\), so we are missing the valuation \(10\) in \((x,y)\).
Hence \(g(t'_1) = g(t'_2) = 1\).

\textbf{Case \(1< g(t'_3) \leq 5n+23\) and \(1< g(t'_4) \leq 5n+23\):}
If \(g(t'_3) < n+4\) and \(g(t'_4)<n+4\), then we will be missing the assignment \(01\) in \((x,z)\). If \(g(t'_3) = n+4 = g(t'_4)\), then we know from the case with visible action that the property does not hold.
If \(n+4 < g(t'_3) < 2n+9\) and \(n+4 < g(t'_4) <2n+9\), then we will be missing the assignment \(01\) in \((x,z)\). 
If either \(g(t'_4)=2n+9\), then \(g(t'_4)[2n+9...(2n+9)+3n+12]=111\) while
\(g(t'_1)[1...3n+13] = 000\) and we will be missing the assignment \(10\) on \((x,z)\). Then, \(g(t'_3)\neq2n+9\) because the suffix of the trace starts with \(000\), so there will be not enough traces to cover for observation \(111\). The same reasoning holds for the next 3 positions. 
The next \(2n+9\) positions cover the deleted letter from \(t_1'\) and \(t_2'\), while the deleted letter from \(t_3'\) and \(t_4'\) happens in a earlier part of the trace. So, the position \(2n+11\) of  \(t_1'\) and \(t_2'\), with assignment \(110\), will miss the assignment \(11\) on \((x,z)\). Note that at that point in the slice of \(t_3'\) and \(t_4'\) \(z\) is constantly \(1\).
\end{proof}

\subsubsection{Lemma \ref{lemma:nredundant:traces_aync}. } 
\emph{For all assignments \(\traceAssign\) over \(\setTraces_n^{\async}\), the valuation at \(2n+11\) is \(n\)-redundant in the trace \(\flatT{\traceAssign}\).}

\begin{proof}
We prove this by induction on the size of trace assignments \(\traceAssign\) over \(\setTraces_n^{\async}\).
\begin{description}
\item[Base case \(|\traceAssign| = 1\):] 
Wlog, let \(\Var(\traceAssign) = \{\pi\}\) for some \(\pi \in \Var\).

If \(\traceAssign(\pi) \in \{t_1, t_2\}\), then at \(2n+11\) we have the block \((1001)^{n+4}\). 
Hence \(t_1[2n+11] = t_1[2n+11+j]\) for all \(1 \leq j \leq n+1\).

If \(\traceAssign(\pi) \in \{t_3, t_4\}\), then at \(2n+10\) we have the block \((1000)^{n+4}\). 
Hence \(t_2[2n+11] = t_1[2n+11+j]\) for all \(1 \leq j \leq n+1\).
%
%

\item[Inductive case:] Assume as induction hypothesis (IH) that the statement holds for all assignments of size \(k\).

Consider an arbitrary assignment \(\traceAssign_{k+1}\) with size \(k+1\)
Then, there exists an assignment \(\traceAssign_{k}\) with size \(k\)
s.t.\ 
\(\traceAssign_{k+1} = \traceAssign_{k}[\pi \mapsto \trace]\) and \(\traceAssign_{k}(\pi)\) is undefined, for some
\(\pi \in \Var\) and \(\trace \in \setTraces_n^{\async}\).
By (IH), the valuation at position \(2n+11\) in \(\flatT{\traceAssign_{k}}\) is \(n\)-redundant. As argued in the base case, the letter at position \(2n+11\) for all 
\(\trace \in \setTraces_n^{\async}\) is \(n\)-redundant, as well.

As \(\traceAssign_{k}(\pi)\) is undefined, then \(\flatT{\traceAssign_{n+1}[\pi \mapsto \trace]} = \flatT{\traceAssign_{n}}\otimes \flatT{\emptyAssign[\pi \mapsto \trace]}\) where \(\otimes\) is the composition of traces.
Then, by the \(2n+1\) letter being \(n\)-redundant in both \(\flatT{\traceAssign_{n}}\) and \(\trace\), it follows that
the letter at \(2n+11\) in \(\flatT{\traceAssign_{n+1}[\pi \mapsto \trace]}\) is \(n\)-redundant,as well.
\qedhere
\end{description}
\end{proof}

\end{document}